\documentclass[lettersize,journal]{IEEEtran}
\usepackage{cite}
\usepackage{amsmath,amssymb,amsfonts}
\usepackage{amsthm}
\usepackage{algorithm}
\usepackage{algpseudocode}
\usepackage{graphicx}
\usepackage{textcomp}
\usepackage{xcolor}
\usepackage{hyperref}
\usepackage{booktabs}
\usepackage{multirow}
\usepackage[table]{xcolor} 
\usepackage{colortbl}
\usepackage{subfig, siunitx, adjustbox}
\usepackage{array}
\usepackage{url}
\usepackage{adjustbox}
\usepackage{verbatim}
\usepackage{graphicx}
\usepackage{cite}
\usepackage{bm}
\usepackage{tabularx}

\newtheorem{theorem}{Theorem}
\newtheorem{lemma}{Lemma}

\begin{document}

\title{Self-Supervised Pretraining for Transmission-Relevant\\MIMO CSI Representation}

\author{Sehyun Ryu, Yumin Kim, Minjae Lee, and Hyun Jong Yang,~\IEEEmembership{Senior Member,~IEEE}, \\John M. Cioffi,~\IEEEmembership{Life Fellow,~IEEE}

\thanks{Corresponding authors: H. J. Yang} 
\thanks{S. Ryu is with the Department of Electrical Engineering, Pohang University of Science and Technology, Pohang, Republic of Korea, and also with the Institute of New Media and Communications, Seoul National University, Seoul, Republic of Korea (e-mail: sh.ryu@postech.ac.kr).}
\thanks{Y. Kim, M. Lee, and H. J. Yang are with the Department of Electrical and Computer Engineering, Seoul National University, Seoul, Republic of Korea. H. J. Yang is also with the Institute of New Media and Communications, Seoul National University (e-mail: 
yumin0107@snu.ac.kr; lmjlmj2580@snu.ac.kr; hjyang@snu.ac.kr).}
\thanks{J. M. Cioffi is with the Department of Electrical Engineering, Stanford University, Stanford, CA 94305 USA (e-mail: cioffi@stanford.edu).}
}

\markboth{Submitted to IEEE Transactions on Wireless Communications}{Ryu \MakeLowercase{\textit{et al.}}: Beam-Response Contrastive Learning}


\maketitle

\begin{abstract}
Self-supervised representation learning from unlabeled channel state information (CSI) can reduce labeling and adaptation overhead in learning-based multiple-input multiple-output (MIMO) systems.
Existing CSI pretraining methods typically use reconstruction objectives or contrastive pairs from generic augmentations, which do not explicitly reflect transmission-relevant channel similarity.
This paper proposes beam-response contrastive learning (BRCL), a self-supervised CSI pretraining framework based on the transmit-side Gram matrix.
For a channel matrix $\mathbf{H}$, $\mathbf{R}=\mathbf{H}^{\mathrm{H}}\mathbf{H}$ determines the received power of any unit-norm transmit beam $\mathbf{w}$ as $|\mathbf{H}\mathbf{w}|_2^2=\mathbf{w}^{\mathrm{H}}\mathbf{R}\mathbf{w}$.
BRCL maps each CSI sample to a beam-response profile and uses the induced soft similarity as a label-free relational target for contrastive pretraining.
Combined with reconstruction learning, BRCL enforces both sample-level CSI recovery and beam-response-level consistency, yielding transferable CSI representations without task-specific labels or manual positive pairs.
Experiments on diverse MIMO channel datasets show that BRCL improves label efficiency and outperforms autoencoder- and channel-charting-based pretraining across beam selection, user selection, and future beam selection tasks.
\end{abstract}


\section{Introduction}
\label{sec:introduction}

Self-supervised learning (SSL) has emerged as a central paradigm for learning transferable representations from unlabeled data by deriving supervisory signals from the data~\cite{DOERSCH2015SSL, NOROOZI2016SSL, GIDARIS2018SSL}.
Contrastive learning (CL), in particular, learns representations by comparing samples in an embedding space using objectives such as InfoNCE~\cite{OORD2019INFONCE}.
Methods such as SimCLR and MoCo showed that two augmented views of the same sample can provide effective positive pairs for representation learning~\cite{CHEN2020SIMCLR, HE2020CONTRAST}.
Accordingly, the learned representation depends strongly on the similarity relation imposed during pretraining.
For wireless channel state information (CSI), this leads to a fundamental question: \emph{what physical relation between CSI samples should be preserved in the representation space?}

Wireless SSL methods differ primarily in how this relation is defined.
One line constructs positive pairs through perturbations or paired observations.
Time--frequency augmentations, multi-receiver observations, and CSI-domain transformations have been adopted for wireless sensing and activity recognition~\cite{LIU2021SSLW, XU2022SSL, YANG2023AUTOFI}.
Self-supervised invariant representations have also been learned from unlabeled CSI and adapted to localization with limited labels~\cite{SALIHU2023SSLWL}.
For beam management, cross-band correspondence has been exploited to support mmWave beam decisions using sub-6 GHz channels~\cite{CHAFAA2022SSLW}.
These methods demonstrate the value of wireless CL when an appropriate invariance or cross-view correspondence is available.
However, the resulting similarity is determined mainly by the selected augmentation, observation pairing, or application-specific assumption, rather than by the communication behavior of distinct CSI samples.

A second line of work uses masked modeling or reconstruction-based pretext tasks.
Motivated by masked language and image modeling~\cite{MELAMUD2016C2V, DELVIN2019BERT, HE2022IMGMSK, BAO2022BEIT}, these methods treat CSI, channel impulse response (CIR), or other radio measurements as structured arrays whose missing components are to be recovered.
Masked prediction has been applied to channel extrapolation, wireless positioning, and radio foundation modeling~\cite{WU2024MASKED, OTT2024RFM, WANG2025MASKED, GULER2025MASKED, LIU2026MASKED}, as well as to multiple-input multiple-output (MIMO) channel feedback, interpolation, and prediction~\cite{GAO2025SSNET, BANERJEE2025MASKED, JIANG2026MASKED}.
By design, however, reconstruction losses primarily reward raw CSI fidelity.
Such fidelity is not always aligned with transmitter-side decisions: CSI samples with different entry-wise reconstruction errors may produce similar beam responses, whereas samples that are close under a reconstruction loss may lead to different precoding behavior.

\def\BRYes{O}
\def\BRNo{X}
\def\BRPartial{\ensuremath{\triangle}}

\begin{table*}[ht]
\centering
\caption{Comparison of representative CSI representation SSL approaches.}
\label{tab:scope_comparison}
\renewcommand{\arraystretch}{1.08}
\setlength{\tabcolsep}{3.0pt}
\scriptsize
\begin{tabular}{@{}p{0.235\textwidth}p{0.265\textwidth}ccccc@{}}
\hline
\multirow{1}{*}[3.6pt]{\textbf{Research Line}}
& \multirow{1}{*}[3.6pt]{\textbf{Supervisory Signal / Similarity}}
& \textbf{\shortstack{Self-\\Supervised}}
& \textbf{\shortstack{Comm.\\Tasks}}
& \textbf{\shortstack{Distinct-CSI\\Relation}}
& \textbf{\shortstack{Beam-Response\\Relation}}
& \textbf{\shortstack{Soft\\Target}}
\tabularnewline
\hline

Augmented-View CL~\cite{LIU2021SSLW, XU2022SSL, YANG2023AUTOFI}
& Hard positives from augmented CSI views
& \BRYes & \BRPartial & \BRPartial & \BRNo & \BRNo
\tabularnewline

Task-Oriented SSL~\cite{SALIHU2023SSLWL, CHAFAA2022SSLW}
& Task-specific invariance or cross-band objective
& \BRYes & \BRPartial & \BRPartial & \BRNo & \BRNo
\tabularnewline

Masked/Reconstruction SSL~\cite{WU2024MASKED, OTT2024RFM, WANG2025MASKED, GULER2025MASKED, LIU2026MASKED, GAO2025SSNET, BANERJEE2025MASKED, JIANG2026MASKED}
& Masked CSI/CIR/radio recovery
& \BRYes & \BRPartial & \BRNo & \BRNo & \BRNo
\tabularnewline

Wireless Foundation Models~\cite{JIAO2024WFM, JIANG2025WFM, GULER2025MASKED, ABOULFOTOUH2025WFMICC, ABOULFOTOUH2025WFM, JIANG2026WFM, XU2025WIFISSL}
& Multimodal alignment or masked recovery
& \BRYes & \BRPartial & \BRPartial & \BRNo & \BRPartial
\tabularnewline

Wi-Fi CSI Sensing SSL~\cite{SHEN2024WIFISSL, GAO2024WIFISSL, LIU2026WFMWIFI, LIU2026WISCAL}
& Sensing-specific invariances
& \BRYes & \BRNo & \BRPartial & \BRNo & \BRNo
\tabularnewline

Channel Charting~\cite{STUDER2018CCHT, FERRAND2021TRIPCCHT, MARGOAROU2021, FERRAND2023MAG, CHAAYA2024CSI}
& Radio-geometry or topology relation
& \BRYes & \BRPartial & \BRYes & \BRNo & \BRPartial
\tabularnewline

\hline
\textbf{BRCL}
& Beam-response soft relation
& \BRYes & \BRYes & \BRYes & \BRYes & \BRYes
\tabularnewline
\hline
\multicolumn{7}{l}{O, $\triangle$, and X denote yes, partial, and no, respectively. Here, ``partial'' means that the property is addressed only in some works or only indirectly.}
\end{tabular}
\vspace{-10pt}
\end{table*}

Recent wireless foundation models extend these ideas toward representations that can be reused across scenarios and tasks.
Multi-modal pretraining has aligned CSI with physical or scattering-environment information~\cite{JIAO2024WFM}, while CIR--CSI consistency has been used to learn shared representations for positioning, beam management, and channel identification~\cite{JIANG2025WFM}.
Other studies combine masked reconstruction and CL to obtain task-agnostic channel representations or unified radio representations for sensing, communication, and localization~\cite{GULER2025MASKED, ABOULFOTOUH2025WFMICC, ABOULFOTOUH2025WFM, JIANG2026WFM, XU2025WIFISSL}.
In Wi-Fi CSI sensing, related SSL and contrastive objectives have improved representation transfer, label efficiency, and few-shot generalization~\cite{SHEN2024WIFISSL, GAO2024WIFISSL, LIU2026WFMWIFI, LIU2026WISCAL}.
These studies show growing interest in reusable wireless representations, but their supervisory signals still arise mainly from reconstruction, cross-modal consistency, or application-specific invariance.
Consequently, the relation that should organize distinct CSI samples for communication-oriented downstream tasks remains largely implicit.

Channel charting provides a complementary approach in which the inter-sample relation is defined more explicitly.
It maps high-dimensional CSI into a low-dimensional chart while preserving local radio geometry without explicit position labels~\cite{STUDER2018CCHT}.
Subsequent methods introduced temporal triplet constraints, phase-insensitive distances, and predictive latent dynamics to organize or predict channel evolution~\cite{FERRAND2021TRIPCCHT, MARGOAROU2021, FERRAND2023MAG, CHAAYA2024CSI}.
The resulting similarity, however, is primarily geometric or topological.
For transmitter-side decisions, geometric proximity does not necessarily imply similar beamforming behavior.
Nearby users may experience different blockage, angular spreads, or path gains, whereas geographically separated users may produce similar responses to a transmit beam codebook.

Taken together, existing methods organize CSI representations through augmentation-induced invariance, reconstruction fidelity, cross-view consistency, or radio geometry.
What remains less explored is a similarity relation defined directly by the physical quantities that govern transmitter-side communication decisions.
To address this gap, this paper takes a precoding-oriented view of self-supervised CSI representation learning.
This perspective is not intended to provide a universal similarity for every wireless task.
Instead, it focuses on AI-based MIMO tasks in which CSI representations support transmitter-side decisions, including beam selection, precoder design, user selection, and beam-response prediction.

For a MIMO channel matrix $\mathbf{H}$ and a transmit beam $\mathbf{w}$, the beamformed received power is given by
\begin{equation}
    \|\mathbf{H}\mathbf{w}\|_2^2
    =
    \mathbf{w}^{\mathrm{H}}
    \mathbf{H}^{\mathrm{H}}
    \mathbf{H}
    \mathbf{w}.
\end{equation}
The transmit-side Gram matrix $\mathbf{H}^{\mathrm{H}}\mathbf{H}$ therefore determines the channel response to any transmit beam and provides a direct physical basis for precoding-relevant CSI similarity~\cite{TELATAR1999}.
This view is consistent with classical MIMO and massive MIMO principles, in which channel Gram matrices, subspaces, and eigenspaces govern multicarrier MIMO beamforming, limited-feedback precoding, reduced-dimensional CSIT, and user grouping~\cite{PALOMAR2003JOINT, LOVE2008LMF, ADHIKARY2013JSDM}.
Accordingly, the frequency-resolved transmit-side Gram-matrix profile serves as a self-supervised relational target for structuring the CSI representation space.

Based on this principle, this paper proposes \emph{beam-response contrastive learning} (BRCL), a self-supervised framework for MIMO orthogonal frequency-division multiplexing (OFDM) CSI representation learning.
BRCL computes a beam-response dissimilarity between CSI samples from their transmit-side Gram-matrix profiles and converts the resulting pairwise relations into a soft target distribution over each mini-batch.
Within the conventional two-view CL framework, augmented views of the same CSI sample remain explicit positive pairs, while distinct samples are weighted according to the similarity of their transmit beam-response profiles.
BRCL thus retains augmentation-induced instance invariance while adding a precoding-relevant relation that is not directly enforced by generic augmentations, masked reconstruction, or geometry-preserving charting objectives.
When combined with reconstruction learning, BRCL jointly promotes sample-level CSI recovery and beam-response-level relational consistency, yielding compact and transferable representations without task-specific labels or manually designed inter-sample positive pairs.
The beam-selection and user-selection tasks assess whether compact CSI representations preserve decision-relevant information, while future beam selection evaluates temporal prediction from past CSI, corresponding to task-oriented CSI compression and channel prediction, respectively~\cite{CHOI2026COMP, LEE2026IMPORTANCE, RYU2026DMRS}.
Table~\ref{tab:scope_comparison} summarizes the supervisory signals and similarity notions used in the related CSI representation-learning approaches.

The main contributions of this paper are:

\begin{itemize}
\vspace{-0.3em}
\item A precoding-oriented CSI similarity is introduced by using the transmit-side Gram-matrix profile.
This similarity compares distinct CSI samples through their induced transmit beam-response profiles and provides a physically meaningful basis for CSI representation learning for transmitter-side decisions.

\item Beam-response contrastive learning (BRCL) is proposed to incorporate this similarity into self-supervised MIMO-OFDM CSI pretraining.
BRCL combines same-instance positive-pair contrastive learning with a soft relational loss over distinct CSI samples, and can be augmented with reconstruction learning to jointly promote instance-level invariance, beam-response-level consistency, and sample-level CSI recovery without task-specific labels or manual positive-pair design.

\item The learned representations are evaluated on beam selection, MU-MIMO user selection, and future beam selection.
Experiments on large-scale ray-tracing MIMO-OFDM datasets show that BRCL improves label efficiency and downstream transfer performance over autoencoder, CL, masked-modeling, and channel-charting-based pretraining baselines.
\end{itemize}

The rest of this paper is organized as follows. Section~\ref{sec:system_model} introduces the MIMO-OFDM channel model and formulates the self-supervised CSI representation learning problem. Section~\ref{sec:brcl} presents the proposed BRCL framework. Section~\ref{sec:theoretical_interpretation} provides the theoretical interpretation of the Gram-matrix-induced target. Section~\ref{sec:experiments} describes the experimental setup and presents the results. Finally, Section~\ref{sec:conclusion} concludes the paper.


\begin{figure*}
\centering
    \includegraphics[width=0.99\linewidth]{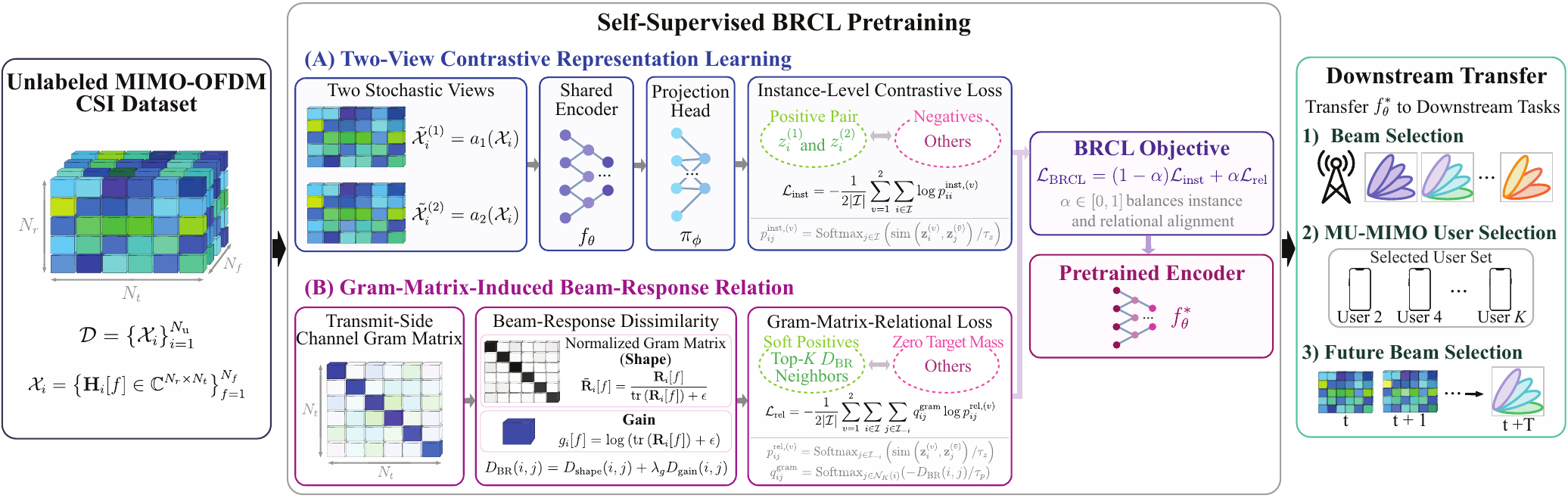}
    \caption{
    Overview of the proposed BRCL framework. BRCL pretrains a CSI encoder by combining two-view instance-level contrastive learning with Gram-matrix-induced relational learning. For each CSI sample, the transmit-side Gram matrix is used to form beam-response profiles, which are separated into normalized beam-response shape and gain components to define the beam-response dissimilarity $D_{\mathrm{BR}}$. The resulting dissimilarities provide soft relational targets, and the pretrained encoder $f_{\theta}^{*}$ is transferred to downstream wireless tasks.
    }
    \label{fig1}
    \vspace{-8pt}
\end{figure*}

\section{System Model and Problem Formulation}
\label{sec:system_model}

\subsection{MIMO-OFDM Channel Dataset}
\label{subsec:channel_dataset}

A downlink MIMO-OFDM system with $N_t$ transmit antennas and $N_r$ receive antennas is considered.
An unlabeled CSI dataset is assumed to be available for self-supervised pretraining:
\begin{equation}
    \mathcal{D}
    =
    \left\{
    \mathcal{X}_i
    \right\}_{i=1}^{N_{\mathrm{u}}},
    \label{eq:unlabeled_dataset}
\end{equation}
where $i$ denotes the sample index and $N_{\mathrm{u}}$ is the number of unlabeled CSI samples.

Each sample $\mathcal{X}_i$ represents an instantaneous MIMO-OFDM channel realization observed over $N_f$ subcarriers:
\begin{equation}
    \mathcal{X}_i
    =
    \left\{
    \mathbf{H}_i[f]
    \in
    \mathbb{C}^{N_r\times N_t}
    \right\}_{f=1}^{N_f},
    \label{eq:csi_sample}
\end{equation}
where $\mathbf{H}_i[f]$ is the channel matrix of the $i$-th sample at subcarrier $f$.
The narrowband flat-fading case is obtained by setting $N_f=1$, in which case a sample reduces to a single channel matrix $\mathbf{H}_i\in\mathbb{C}^{N_r\times N_t}$.

At subcarrier $f$, the received signal is modeled as
\begin{equation}
    \mathbf{y}_i[f]
    =
    \mathbf{H}_i[f]\mathbf{x}_i[f]
    +
    \mathbf{n}_i[f],
    \label{eq:received_signal}
\end{equation}
where $\mathbf{x}_i[f]\in\mathbb{C}^{N_t}$ is the transmitted signal vector and $\mathbf{n}_i[f]$ is additive noise. The CSI sample $\mathcal{X}_i$ is used as the input to the learning model.
In implementation, the complex-valued CSI can be represented by stacking its real and imaginary components, but this preprocessing step is omitted from the notation for clarity.

\subsection{Self-Supervised CSI Representation Learning}
\label{subsec:ssl_problem}

The objective of self-supervised CSI representation learning is to pretrain an encoder from the unlabeled dataset $\mathcal{D}$ without using task-specific labels. Let $f_{\theta}(\cdot)$ denote the CSI encoder that maps a CSI sample to a $d_{\mathrm{rep}}$-dimensional representation. For each CSI sample $\mathcal{X}_i\in\mathcal{D}$, the encoder produces
\begin{equation}
    \mathbf{u}_i
    =
    f_{\theta}(\mathcal{X}_i)
    \in
    \mathbb{R}^{d_{\mathrm{rep}}},
    \label{eq:encoder_output}
\end{equation}
where $\mathbf{u}_i$ denotes the learned representation of $\mathcal{X}_i$.
The notation $\mathbf{u}_i$ is adopted for the representation to avoid confusion with the channel matrix $\mathbf{H}_i[f]$.

After pretraining, the encoder is transferred to a downstream wireless task using a labeled dataset
\begin{equation}
    \mathcal{S}^{(r)}
    =
    \left\{
    \left(
    \mathcal{X}_j,
    \eta_j^{(r)}
    \right)
    \right\}_{j=1}^{|\mathcal{S}^{(r)}|},
    \label{eq:downstream_dataset}
\end{equation}
where $r$ denotes the downstream task index and $\eta_j^{(r)}$ is the corresponding task-specific target.
Depending on the task, $\eta_j^{(r)}$ may denote a codebook beam label, a MU-MIMO user-set utility or quality indicator, or a future beam-response target such as a future beam index or beam-gain profile.

A task-specific head $\ell_{\omega}^{(r)}(\cdot)$ is trained on top of the pretrained encoder. For a downstream CSI sample $\mathcal{X}_j\in\mathcal{S}^{(r)}$, the predicted target is
\begin{equation}
    \hat{\eta}_j^{(r)}
    =
    \ell_{\omega}^{(r)}
    \left(
    f_{\theta}
    \left(
    \mathcal{X}_j
    \right)
    \right).
    \label{eq:downstream_prediction}
\end{equation}
The downstream performance is evaluated by comparing $\hat{\eta}_j^{(r)}$ with the ground-truth target $\eta_j^{(r)}$ on held-out samples from $\mathcal{S}^{(r)}$.
Depending on the evaluation protocol, the encoder $f_{\theta}$ is either frozen to assess representation quality or fine-tuned jointly with the task head.
The effectiveness of the pretrained encoder is measured by its transferability and label efficiency across downstream tasks.

\subsection{Transmit-Response-Oriented Problem Formulation}
\label{subsec:tx_response_problem}

A key question in self-supervised CSI representation learning is which physical relation among unlabeled CSI samples should be preserved in the representation space.
Existing self-supervised approaches often treat CSI as a generic data array and define pretext objectives through raw channel reconstruction, masking, or augmentation-based instance discrimination.
While such objectives can learn useful features, they do not explicitly specify which channel properties should be preserved for communication decisions.
This work instead targets \emph{transmit-domain CSI representations}, where the relevant information is how the channel responds to candidate beams, precoders, and scheduled user groups, rather than how accurately the raw CSI entries can be reconstructed.
Accordingly, the representation is not designed to reconstruct the full CSI sample, but to preserve the information that determines transmit responses over OFDM subcarriers.

For a unit-norm transmit beam $\mathbf{w}\in\mathbb{C}^{N_t}$, the received power of CSI sample $\mathcal{X}_i$ at subcarrier $f$ is
\begin{equation}
    \gamma_i[f;\mathbf{w}]
    =
    \left\|
    \mathbf{H}_i[f]\mathbf{w}
    \right\|_2^2
    =
    \mathbf{w}^{H}
    \mathbf{R}_i[f]
    \mathbf{w},
    \label{eq:subcarrier_beam_gain}
\end{equation}
where
\begin{equation}
    \mathbf{R}_i[f]
    =
    \mathbf{H}_i[f]^{H}
    \mathbf{H}_i[f]
    \in
    \mathbb{C}^{N_t\times N_t}
    \label{eq:subcarrier_tx_covariance}
\end{equation}
is the transmit-side channel Gram matrix of sample $\mathcal{X}_i$ at subcarrier $f$.

Together, \eqref{eq:subcarrier_beam_gain} and \eqref{eq:subcarrier_tx_covariance} show that $\mathbf{R}_i[f]$ determines the transmit beam response of the channel at subcarrier $f$.
Accordingly, the frequency-resolved beam-response profile of sample $\mathcal{X}_i$ is defined as
\begin{equation}
    \mathcal{B}_i
    =
    \left\{
    \mathbf{R}_i[f]
    \right\}_{f=1}^{N_f}.
    \label{eq:beam_response_profile}
\end{equation}
The profile $\mathcal{B}_i$ characterizes the transmitter-side response of an instantaneous MIMO-OFDM CSI sample across frequency.
Since different subcarriers may exhibit different gains, ranks, and preferred transmit directions, preserving this frequency-resolved response is important for OFDM precoding, power allocation, and feedback-related decisions.

The proposed problem is to learn an encoder $f_{\theta}$ from the unlabeled dataset $\mathcal{D}$ such that the representation $\mathbf{u}_i=f_{\theta}(\mathcal{X}_i)$ preserves the transmit-response geometry induced by $\mathcal{B}_i$. In this formulation, two CSI samples should be close in the learned representation space when their beam-response profiles induce similar beam gains over subcarriers and transmit beams. Thus, similarity in $\mathcal{B}_i$ implies that the transmitter can treat the corresponding channels similarly for beamforming-oriented decisions, even without reconstructing the full CSI samples. BRCL exploits this transmit-response relation to learn compact, label-free, and communication-relevant CSI representations.

\section{Beam-Response Contrastive Learning}
\label{sec:brcl}

\begin{figure}
\centering
    \includegraphics[width=0.75\linewidth]{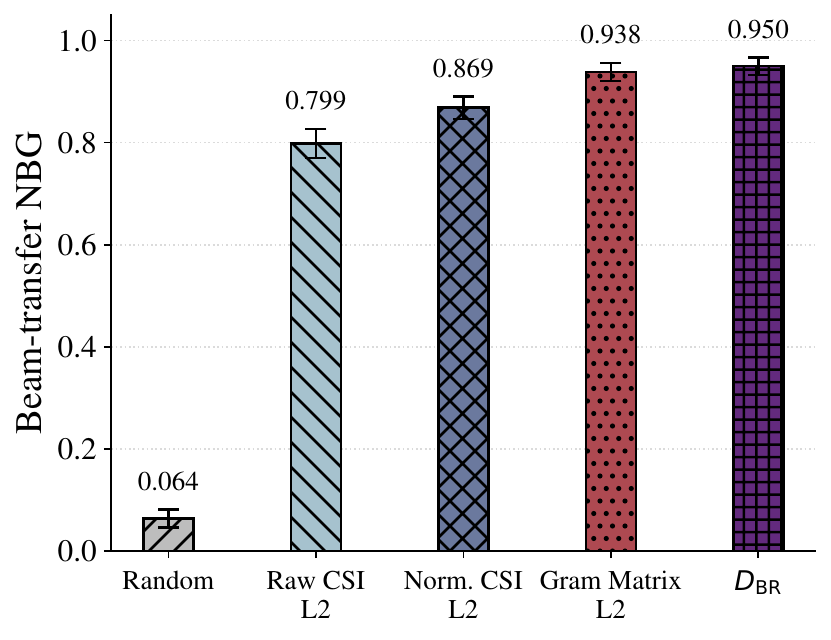}
    \vspace{-7pt}
    \caption{
    Beam-transfer NBG of nearest-neighbor CSI samples under different similarity metrics, evaluated on the DeepMIMO dataset.
    }
    \label{fig_motivation}
    \vspace{-9pt}
\end{figure}

This section presents the proposed BRCL framework.
BRCL aims to learn CSI embeddings whose geometry reflects the transmit-side physical behavior of MIMO-OFDM channels.
Instead of defining sample similarity solely through generic augmentations or channel-domain reconstruction, BRCL constructs a soft relational target from the transmit-side Gram-matrix profile of each channel realization, which characterizes CSI similarity via subcarrier-wise beam responses.
To this end, BRCL combines instance-level contrastive learning with physically grounded relational learning.
Augmented views of the same CSI sample are aligned to ensure representation invariance, while CSI samples with similar Gram-matrix-induced beam-response profiles are encouraged to have nearby embeddings through a Gram-matrix-relational objective.
An optional $\ell_2$ reconstruction loss can further be incorporated to preserve channel-domain information.
As a result, BRCL retains the stability of instance discrimination while organizing the embedding space according to transmission-relevant channel similarity.

\subsection{CSI Augmentation and Encoder}
\label{subsec:encoder_aug}

For each CSI sample $\mathcal{X}_i\in\mathcal{D}$, two stochastic views are generated as
\begin{equation}
    \tilde{\mathcal{X}}_i^{(1)}
    =
    a_1(\mathcal{X}_i),
    \qquad
    \tilde{\mathcal{X}}_i^{(2)}
    =
    a_2(\mathcal{X}_i),
    \label{eq:augmented_views}
\end{equation}
where $a_1(\cdot)$ and $a_2(\cdot)$ denote independently sampled CSI augmentations.
This two-view construction follows the standard CL protocol, where independently perturbed views of the same sample form an explicit positive pair~\cite{CHEN2020SIMCLR}.
In the CSI setting, such augmentations can model noisy or partially observed channel measurements, as commonly adopted in wireless SSL through channel perturbations such as noise injection and masking~\cite{GULER2025MASKED}.
The resulting positive pair encourages the encoder to learn representations that are invariant to measurement-level perturbations while preserving the identity of the underlying channel realization.

The physical relation among different CSI samples, however, is not defined by the augmentations. Instead, it is defined by the Gram-matrix-induced beam-response similarity introduced in the next subsection.
Thus, BRCL retains the stability of two-view instance-level CL while injecting transmitter-side physical structure through a soft relational target.

Let $f_{\theta}(\cdot)$ denote the CSI encoder and $\pi_{\phi}(\cdot)$ denote the projection head used during self-supervised pretraining. The encoder maps each augmented view to a representation
\begin{equation}
    \mathbf{u}_i^{(v)}
    =
    f_{\theta}
    \left(
        \tilde{\mathcal{X}}_i^{(v)}
    \right)
    \in
    \mathbb{R}^{d_{\mathrm{rep}}},
    \qquad
    v \in \{1,2\},
    \label{eq:encoder_view_output}
\end{equation}
where $d_{\mathrm{rep}}$ denotes the dimension of the encoder representation. The projection head then maps this representation to a normalized contrastive embedding
\begin{equation}
    \mathbf{z}_i^{(v)}
    =
    \frac{
        \pi_{\phi}
        \left(
            \mathbf{u}_i^{(v)}
        \right)
    }{
        \left\|
            \pi_{\phi}
            \left(
                \mathbf{u}_i^{(v)}
            \right)
        \right\|_2
    }
    \in
    \mathbb{R}^{d_z},
    \qquad
    v \in \{1,2\},
    \label{eq:projected_embedding}
\end{equation}
where $d_z$ denotes the dimension of the projection space in which the contrastive objectives are applied. The projection head is used only during pretraining, and the encoder representation $\mathbf{u}_i^{(v)}$ is used for downstream tasks.

\begin{figure*}[t]
\captionsetup[subfloat]{farskip=2pt}
\vspace{-7pt}
\centering
        \subfloat[SimCLR one-hot target]{\includegraphics[width=.3\linewidth]{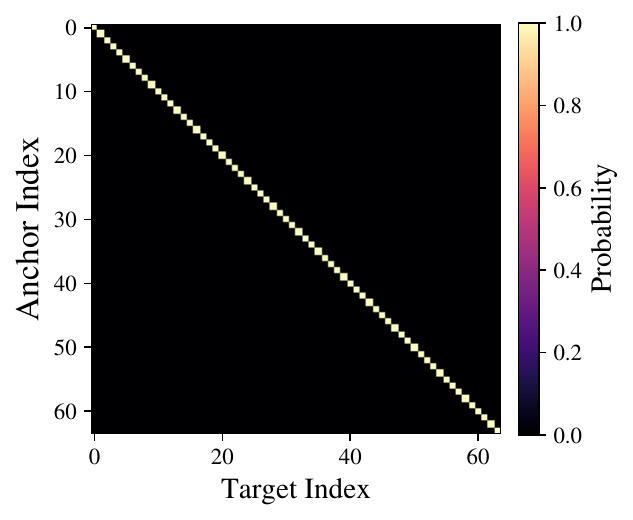}}
        \subfloat[Top-K soft target $q_{ij}^{\mathrm{gram}}$]{\includegraphics[width=.3\linewidth]{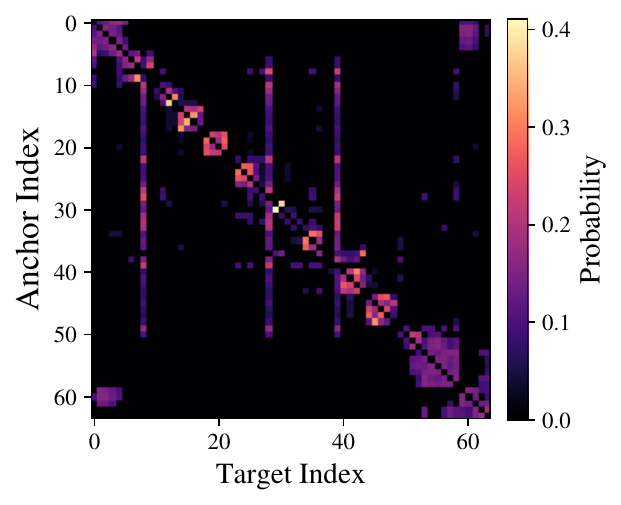}}
        \subfloat[Ranked target probabilities]{\includegraphics[width=.3\linewidth]{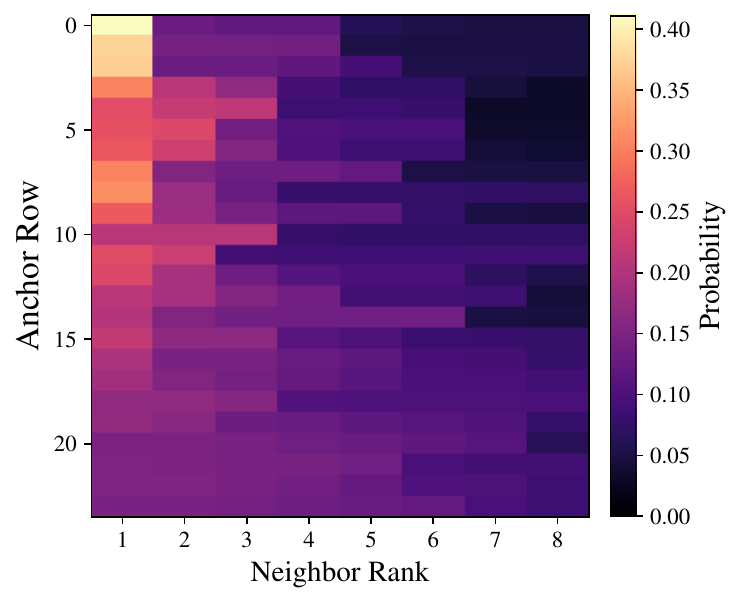}}
        \\
\vspace{-3pt}
\caption{Construction of the Gram-matrix-induced soft target ($q_{ij}^{\mathrm{gram}}$). Unlike the one-hot SimCLR target, BRCL assigns nonzero probability to the top-$K$ beam-response-similar CSI samples ($K=8$ in the figure), providing structured soft supervision over communication-relevant neighbors.
}
\vspace{-7pt}
\label{fig:q_gram}
\end{figure*}

\subsection{Soft Relational Contrastive Target}
\label{subsec:soft_target}


BRCL constructs a physical target distribution from the frequency-resolved beam-response profile defined in \eqref{eq:beam_response_profile}.
The transmit-side Gram matrix is first computed at each subcarrier as
\begin{equation}
\mathbf{R}_i[f]
=
\mathbf{H}_i[f]^{H}
\mathbf{H}_i[f],
\qquad
f=1,\ldots,N_f.
\label{eq:brcl_subcarrier_covariance}
\end{equation}
The corresponding subcarrier beam-gain component is defined as
\begin{equation}
g_i[f]
=
\log
\left(
\operatorname{tr}
\left(
\mathbf{R}_i[f]
\right)
+
\epsilon
\right),
\label{eq:subcarrier_gain_component}
\end{equation}
where $\operatorname{tr}\left(\mathbf{R}_i[f]\right)$ represents the aggregate transmit-side channel gain at subcarrier $f$, and $\epsilon>0$ is a small constant introduced for numerical stability.

To separate the relative beam-response shape from the absolute gain, the normalized Gram matrix is further defined as
\begin{equation}
\bar{\mathbf{R}}_i[f]
=
\frac{
\mathbf{R}_i[f]
}{
\operatorname{tr}
\left(
\mathbf{R}_i[f]
\right)
+
\epsilon
}.
\label{eq:normalized_subcarrier_covariance}
\end{equation}
Accordingly, $\bar{\mathbf{R}}_i[f]$ captures the relative transmit beam-response shape at subcarrier $f$, whereas $g_i[f]$ captures the corresponding absolute beam gain.

Since the dynamic range of the gain component can differ substantially from that of the normalized Gram-matrix shape component, the gain term is calibrated using statistics computed from the unlabeled pretraining dataset. Specifically, for each subcarrier index $f$, let
\begin{align}
    \mu_g[f]
    &=
    \frac{1}{N_u}
    \sum_{\mathcal{X}_n\in\mathcal{D}}
    g_n[f],
    \\
    \sigma_g[f]
    &=
    \left(
    \frac{1}{N_u}
    \sum_{\mathcal{X}_n\in\mathcal{D}}
    \left(
    g_n[f]-\mu_g[f]
    \right)^2
    \right)^{1/2},
    \label{eq:gain_statistics}
\end{align}
where $N_u=|\mathcal{D}|$ is the number of unlabeled CSI samples, and $\mu_g[f]$ and $\sigma_g[f]$ denote the empirical mean and standard deviation of the gain component at subcarrier $f$, respectively. The normalized gain component is then defined as
\begin{equation}
    \tilde{g}_i[f]
    =
    \frac{
    g_i[f]-\mu_g[f]
    }{
    \sigma_g[f]+\epsilon_g
    },
    \label{eq:normalized_gain_component}
\end{equation}
where $\epsilon_g>0$ is a small constant for numerical stability. The statistics $\mu_g[f]$ and $\sigma_g[f]$ are computed once from the unlabeled pretraining dataset and kept fixed throughout pretraining.

For two CSI samples $\mathcal{X}_i$ and $\mathcal{X}_j$, the Gram-matrix-shape dissimilarity is defined as
\begin{equation}
    D_{\mathrm{shape}}(i,j)
    =
    \frac{1}{N_f}
    \sum_{f=1}^{N_f}
    \left\|
    \bar{\mathbf{R}}_i[f]
    -
    \bar{\mathbf{R}}_j[f]
    \right\|_F^2,
    \label{eq:brcl_shape_distance}
\end{equation}
and the normalized gain dissimilarity is defined as
\begin{equation}
    D_{\mathrm{gain}}(i,j)
    =
    \frac{1}{N_f}
    \sum_{f=1}^{N_f}
    \left(
    \tilde{g}_i[f]-\tilde{g}_j[f]
    \right)^2.
    \label{eq:brcl_gain_distance}
\end{equation}
The beam-response dissimilarity is then given by
\begin{equation}
    D_{\mathrm{BR}}(i,j)
    =
    D_{\mathrm{shape}}(i,j)
    +
    \lambda_g
    D_{\mathrm{gain}}(i,j),
    \label{eq:brcl_distance}
\end{equation}
where $\lambda_g\geq 0$ controls the contribution of the normalized absolute gain component. The first term compares trace-normalized Gram matrices and captures the beam-response shape, including dominant transmit directions, rank structure, and relative eigenmode strengths. Since trace normalization removes the overall channel power, the second term reintroduces absolute gain information through $\tilde{g}_i[f]$. This decomposition prevents large-scale power differences from dominating the shape comparison, while incorporating path loss, shadowing, and blockage effects with a controlled scale.

To assess whether $D_{\mathrm{BR}}$ selects beam-decision-relevant neighbors, Fig.~\ref{fig_motivation} compares nearest-neighbor transfer across different CSI similarity metrics. For each anchor sample $i$, the neighbor $j$ selected by each metric provides its optimal beam $k_j^{\star}$, which is evaluated on the anchor channel using the normalized beamforming gain (NBG)
\begin{equation}
    \mathrm{NBG}(i \leftarrow j)
    =
    \frac{
        G_i(k_j^{\star})
    }{
        G_i(k_i^{\star})
    },
    \label{eq:nbg_motivation}
\end{equation}
where $k_i^{\star}$ is the anchor-optimal beam and $G_i(k)$ denotes the beamforming gain of beam $k$ for sample $i$. Thus, $1-\mathrm{NBG}(i \leftarrow j)$ represents the normalized regret of using the neighbor's beam. As shown in Fig.~\ref{fig_motivation}, $D_{\mathrm{BR}}$ achieves the highest NBG among the considered metrics, supporting its use for constructing beam-relevant relational targets in BRCL.

Motivated by this observation, $D_{\mathrm{BR}}$ is used during pretraining to identify local beam-response neighbors within each mini-batch. At each training iteration, a mini-batch is sampled as
\begin{equation}
    \mathcal{M}
    =
    \left\{
    \mathcal{X}_i
    \right\}_{i\in\mathcal{I}}
    \subset
    \mathcal{D},
\end{equation}
where $\mathcal{I}$ denotes the index set of samples included in the mini-batch. For an anchor sample $i\in\mathcal{I}$, let
\begin{equation}
    \mathcal{I}_{-i}
    =
    \mathcal{I}\setminus\{i\}
\end{equation}
denote the set of candidate indices excluding the anchor. The top-$K$ beam-response-neighbor set of anchor sample $i$ is then defined as
\begin{equation}
    \mathcal{N}_{K}(i)
    =
    \operatorname*{arg\,topK}_{j\in\mathcal{I}_{-i}}
    \left(
    -D_{\mathrm{BR}}(i,j)
    \right),
    \label{eq:topk_covariance_neighbor_set}
\end{equation}
where $\operatorname*{arg\,topK}$ returns the $K$ indices with the largest values of $-D_{\mathrm{BR}}(i,j)$, equivalently the $K$ smallest beam-response dissimilarities.

The Gram-matrix-induced soft relational target for anchor sample $i$ is then defined as
\begin{equation}
    q_{ij}^{\mathrm{gram}}
    =
    \begin{cases}
    \displaystyle
    \frac{
    \exp
    \left(
    -D_{\mathrm{BR}}(i,j)/\tau_p
    \right)
    }{
    \sum_{k\in\mathcal{N}_{K}(i)}
    \exp
    \left(
    -D_{\mathrm{BR}}(i,k)/\tau_p
    \right)
    },
    & j\in\mathcal{N}_{K}(i), \\[3ex]
    0,
    & \text{otherwise},
    \end{cases}
    \label{eq:cov_soft_target}
\end{equation}
where $\tau_p>0$ is a temperature that controls the sharpness of the physical similarity distribution induced by the beam-response distance.
The anchor itself is excluded from $q_{ij}^{\mathrm{gram}}$ because the same-instance positive relation is optimized separately through the instance-level contrastive loss.
The top-$K$ restriction assigns nonzero target probability only to the most similar Gram-matrix-induced beam-response neighbors, while the remaining mini-batch samples receive zero target mass and act as negatives through the embedding-space softmax normalization.

The beam-response descriptors can be precomputed for a fixed unlabeled dataset.
During training, BRCL only requires the pairwise beam-response dissimilarities $D_{\mathrm{BR}}(i,j)$ among samples in the mini-batch index set $\mathcal{I}$, and the corresponding top-$K$ Gram-matrix-neighbor sets $\mathcal{N}_K(i)$.
These quantities can be computed efficiently by flattening the normalized Gram matrix profiles across subcarriers, appending the normalized gain components, and applying vectorized matrix operations.


\subsection{Embedding-Space Similarity Distributions}
\label{subsec:embedding_similarity}

Given the normalized embeddings in \eqref{eq:projected_embedding}, BRCL defines two cross-view embedding-space similarity distributions: one for same-instance CL and one for Gram-matrix-relational learning.
Let $v\in\{1,2\}$ denote the anchor view and let $\bar{v}=3-v$ denote the opposite view.

For the instance-level contrastive loss, the probability assigned to candidate sample $j\in\mathcal{I}$ for anchor sample $i\in\mathcal{I}$ is
\begin{equation}
    p_{ij}^{\mathrm{inst},(v)}
    =
    \frac{
    \exp
    \left(
    \operatorname{sim}
    \left(
    \mathbf{z}_i^{(v)},
    \mathbf{z}_j^{(\bar{v})}
    \right)
    /\tau_z
    \right)
    }{
    \sum_{k\in\mathcal{I}}
    \exp
    \left(
    \operatorname{sim}
    \left(
    \mathbf{z}_i^{(v)},
    \mathbf{z}_k^{(\bar{v})}
    \right)
    /\tau_z
    \right)
    },
    \qquad
    j\in\mathcal{I}.
    \label{eq:instance_embedding_distribution}
\end{equation}
The same-instance probability $p_{ii}^{\mathrm{inst},(v)}$ is used to preserve instance-level invariance between the two augmented views of the same CSI sample.

For the Gram-matrix-relational loss, the same-instance candidate is excluded, and the off-diagonal embedding-space distribution is defined as
\begin{equation}
    p_{ij}^{\mathrm{rel},(v)}
    =
    \frac{
    \exp
    \left(
    \operatorname{sim}
    \left(
    \mathbf{z}_i^{(v)},
    \mathbf{z}_j^{(\bar{v})}
    \right)
    /\tau_z
    \right)
    }{
    \sum_{k\in\mathcal{I}_{-i}}
    \exp
    \left(
    \operatorname{sim}
    \left(
    \mathbf{z}_i^{(v)},
    \mathbf{z}_k^{(\bar{v})}
    \right)
    /\tau_z
    \right)
    },
    \qquad
    j\in\mathcal{I}_{-i}.
    \label{eq:relational_embedding_distribution}
\end{equation}
The distribution $p_{ij}^{\mathrm{rel},(v)}$ represents how the learned embedding space ranks distinct CSI samples in the opposite view relative to anchor sample $i$.
Since the embeddings are $\ell_2$-normalized, cosine similarity is equivalent to the inner product.

Fig.~\ref{fig:q_gram} illustrates the Gram-matrix-induced target distribution.
Unlike the one-hot SimCLR target, BRCL assigns soft probabilities to the top-$K$ CSI samples with the smallest beam-response dissimilarities, with larger weights given to more similar neighbors.
This structured target guides the embedding space to preserve communication-relevant inter-sample similarity beyond instance-level discrimination.

\subsection{BRCL Objective}
\label{subsec:brcl_objective}

BRCL optimizes instance-level invariance and Gram-matrix-induced relational learning as two separate objectives. 
This separation prevents the same-instance positive pair from competing with Gram-matrix-neighbor positives in a single mixed target distribution, while retaining both sources of self-supervision.

The instance-level contrastive loss is defined as
\begin{equation}
    \mathcal{L}_{\mathrm{inst}}
    =
    -
    \frac{1}{2|\mathcal{I}|}
    \sum_{v=1}^{2}
    \sum_{i\in\mathcal{I}}
    \log
    p_{ii}^{\mathrm{inst},(v)}.
    \label{eq:instance_loss}
\end{equation}
This term is equivalent to the standard SimCLR-style cross-view instance discrimination objective over the mini-batch.

The Gram-matrix-relational loss aligns the off-diagonal embedding-space distribution $p_{ij}^{\mathrm{rel},(v)}$ with the Gram-matrix-induced target distribution $q_{ij}^{\mathrm{gram}}$:
\begin{equation}
    \mathcal{L}_{\mathrm{rel}}
    =
    -
    \frac{1}{2|\mathcal{I}|}
    \sum_{v=1}^{2}
    \sum_{i\in\mathcal{I}}
    \sum_{j\in\mathcal{I}_{-i}}
    q_{ij}^{\mathrm{gram}}
    \log
    p_{ij}^{\mathrm{rel},(v)}.
    \label{eq:relational_loss}
\end{equation}
Because $q_{ij}^{\mathrm{gram}}$ is nonzero only for $j\in\mathcal{N}_{K}(i)$, only the top-$K$ Gram-matrix-induced beam-response neighbors contribute positive target mass. 
All other off-diagonal samples remain in the denominator of \eqref{eq:relational_embedding_distribution} and therefore act as negatives.

The final BRCL pretraining loss is
\begin{equation}
    \mathcal{L}_{\mathrm{BRCL}}
    =
    (1-\alpha)
    \mathcal{L}_{\mathrm{inst}}
    +
    \alpha
    \mathcal{L}_{\mathrm{rel}},
    \qquad
    0\leq\alpha\leq 1.
    \label{eq:brcl_loss}
\end{equation}
The parameter $\alpha$ controls the trade-off between instance-level invariance and physical relational learning. 
When $\alpha=0$, BRCL reduces to a SimCLR-style instance contrastive objective. 
When $\alpha=1$, the objective is determined entirely by the Gram-matrix-induced beam-response relation among distinct CSI samples. 
For $0<\alpha<1$, the encoder is trained to preserve both same-instance invariance and Gram-matrix-induced transmit-response relations.

This objective extends conventional instance discrimination with a communication-aware relational loss, rather than replacing the one-hot positive pair with a mixed target distribution.

\subsection{Overall Pretraining Objective}
\label{subsec:overall_objective}

To further preserve channel-domain information during pretraining, the BRCL objective can include an optional $\ell_2$ reconstruction regularizer.
Let $r_{\psi}(\cdot)$ denote a decoder used only during pretraining. For each original CSI sample $\mathcal{X}_i$, the encoder and decoder produce
\begin{equation}
    \mathbf{u}_i
    =
    f_{\theta}
    \left(
        \mathcal{X}_i
    \right),
    \qquad
    \hat{\mathcal{X}}_i
    =
    r_{\psi}
    \left(
        \mathbf{u}_i
    \right).
    \label{eq:reconstructed_csi}
\end{equation}
The reconstruction loss for a mini-batch $\mathcal{M}$ is defined as
\begin{equation}
    \mathcal{L}_{\mathrm{rec}}
    =
    \frac{1}{|\mathcal{I}|}
    \sum_{i\in\mathcal{I}}
    \left\|
        \hat{\mathcal{X}}_i
        -
        \mathcal{X}_i
    \right\|_F^2,
    \label{eq:reconstruction_loss}
\end{equation}
where $\|\cdot\|_F$ denotes the Frobenius norm. The overall pretraining loss is then given by
\begin{equation}
    \mathcal{L}_{\mathrm{pre}}
    =
    \mathcal{L}_{\mathrm{BRCL}}
    +
    \beta_{\mathrm{rec}}
    \mathcal{L}_{\mathrm{rec}},
    \label{eq:overall_pretraining_loss}
\end{equation}
where $\beta_{\mathrm{rec}}\geq 0$ controls the contribution of the reconstruction regularizer. Setting $\beta_{\mathrm{rec}}=0$ recovers the BRCL objective without reconstruction.

The encoder, projection head, and optional decoder are pretrained by minimizing the overall loss over mini-batches sampled from the unlabeled CSI dataset:
\begin{equation}
    (\theta^{\star},\phi^{\star},\psi^{\star})
    =
    \arg\min_{\theta,\phi,\psi}
    \;
    \mathbb{E}_{\mathcal{M}\sim\mathcal{D}}
    \left[
        \mathcal{L}_{\mathrm{pre}}
        \left(
            \mathcal{M};\theta,\phi,\psi
        \right)
    \right].
    \label{eq:overall_pretraining_objective}
\end{equation}
For each mini-batch, the Gram-matrix-induced target $q_{ij}^{\mathrm{gram}}$ is constructed from the beam-response dissimilarity $D_{\mathrm{BR}}(i,j)$ and treated as a parameter-independent physical target during backpropagation. The objective does not exploit task-specific labels. After pretraining, the projection head $\pi_{\phi}(\cdot)$ and the decoder $r_{\psi}(\cdot)$, if used, are discarded, and the pretrained encoder $f_{\theta^{\star}}(\cdot)$ is used to extract the transferable representation of a downstream CSI sample $\mathcal{X}$:
\begin{equation}
    \mathbf{u}
    =
    f_{\theta^{\star}}(\mathcal{X})
    \in
    \mathbb{R}^{d_{\mathrm{rep}}}.
    \label{eq:pretrained_representation}
\end{equation}



\begin{figure}[t]
\captionsetup[subfloat]{farskip=2pt}
\vspace{-7pt}
\centering
\subfloat[Global random pairs]{\includegraphics[width=.88\linewidth]{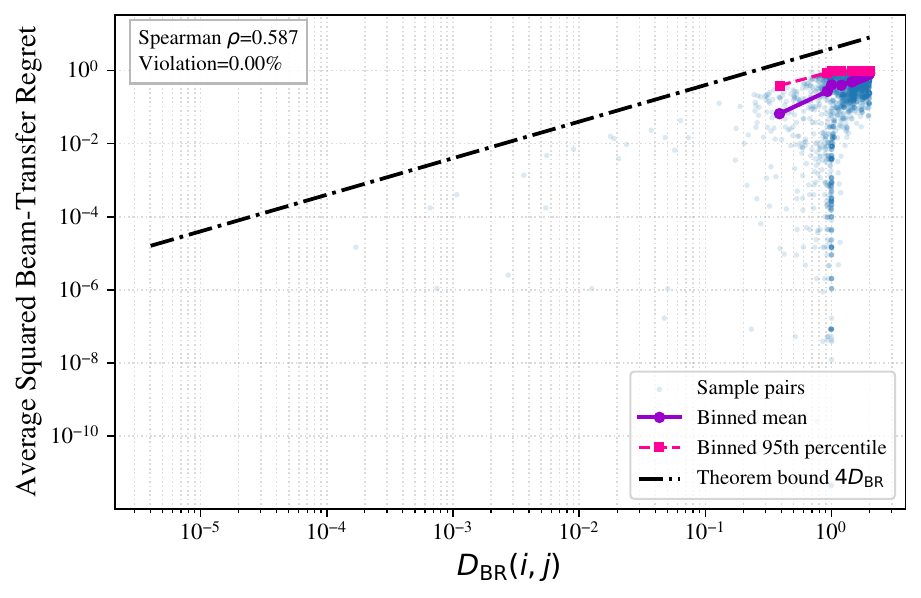}}\
\subfloat[Top-16 local $D_{\mathrm{BR}}$-neighbor pairs]{\includegraphics[width=.88\linewidth]{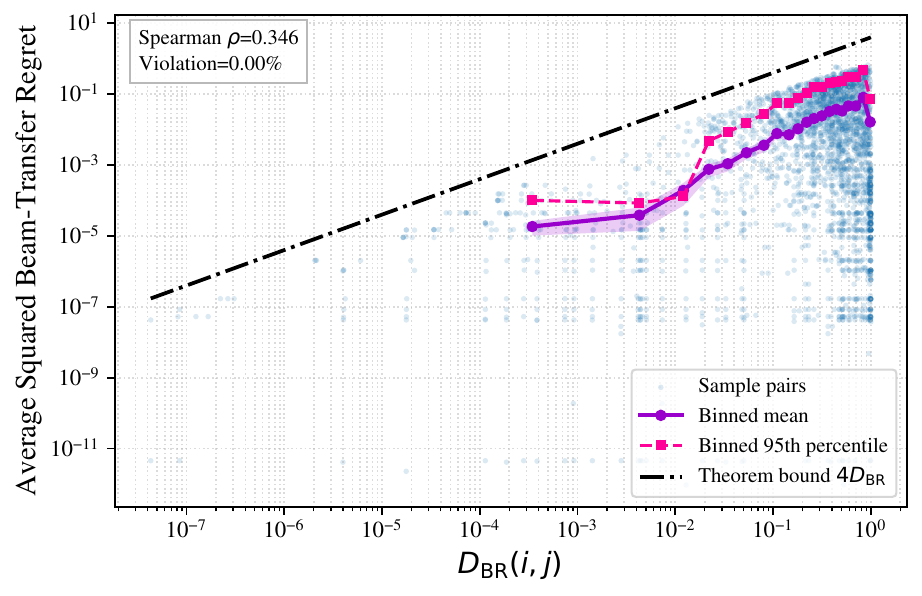}}
\vspace{-3pt}
\caption{Empirical evaluation of Theorem~\ref{thm:beam_decision_regret}. The average squared beam-transfer regret is plotted against $D_{\mathrm{BR}}(i,j)$ for global random CSI pairs and top-16 local $D_{\mathrm{BR}}$-neighbor pairs, with binned curves computed using 5\%-wide $D_{\mathrm{BR}}$ bins.}
\label{fig:dbr_regret_bound}
\vspace{-6pt}
\end{figure}

\section{Theoretical Interpretation}
\label{sec:theoretical_interpretation}

This section provides a physical interpretation of the Gram-matrix-induced target used in BRCL. 
The purpose of the analysis is not to establish an end-to-end performance guarantee for the learned downstream predictors.
Rather, the Gram-matrix-shape component of $D_{\mathrm{BR}}(i,j)$ is shown to upper-bound discrepancies between normalized transmit beam responses.
This result explains why $D_{\mathrm{BR}}(i,j)$ is a meaningful self-supervised relation for transmitter-side CSI representation learning.

For a unit-norm transmit beam $\mathbf{w}$, define the normalized transmit beam response of sample $\mathcal{X}_i$ at subcarrier $f$ as
\begin{equation}
    \bar{\gamma}_i[f;\mathbf{w}]
    =
    \mathbf{w}^{H}\bar{\mathbf{R}}_i[f]\mathbf{w}.
    \label{eq:normalized_beam_response}
\end{equation}
For notational compactness, define the Gram-matrix-shape difference between samples $\mathcal{X}_i$ and $\mathcal{X}_j$ as
\begin{equation}
    \mathbf{A}_{ij}[f]
    =
    \bar{\mathbf{R}}_i[f]
    -
    \bar{\mathbf{R}}_j[f].
    \label{eq:gram_shape_difference}
\end{equation}
Since $\bar{\mathbf{R}}_i[f]$ and $\bar{\mathbf{R}}_j[f]$ are Hermitian, $\mathbf{A}_{ij}[f]$ is also Hermitian.
The following results use the shape component $D_{\mathrm{shape}}(i,j)$.
The gain component in $D_{\mathrm{BR}}(i,j)$ is not required for the normalized-response bounds; since $\lambda_g\geq 0$, it follows that $D_{\mathrm{shape}}(i,j)\leq D_{\mathrm{BR}}(i,j)$.

\begin{figure*}[ht]
    \centering
    \includegraphics[width=0.97\linewidth]{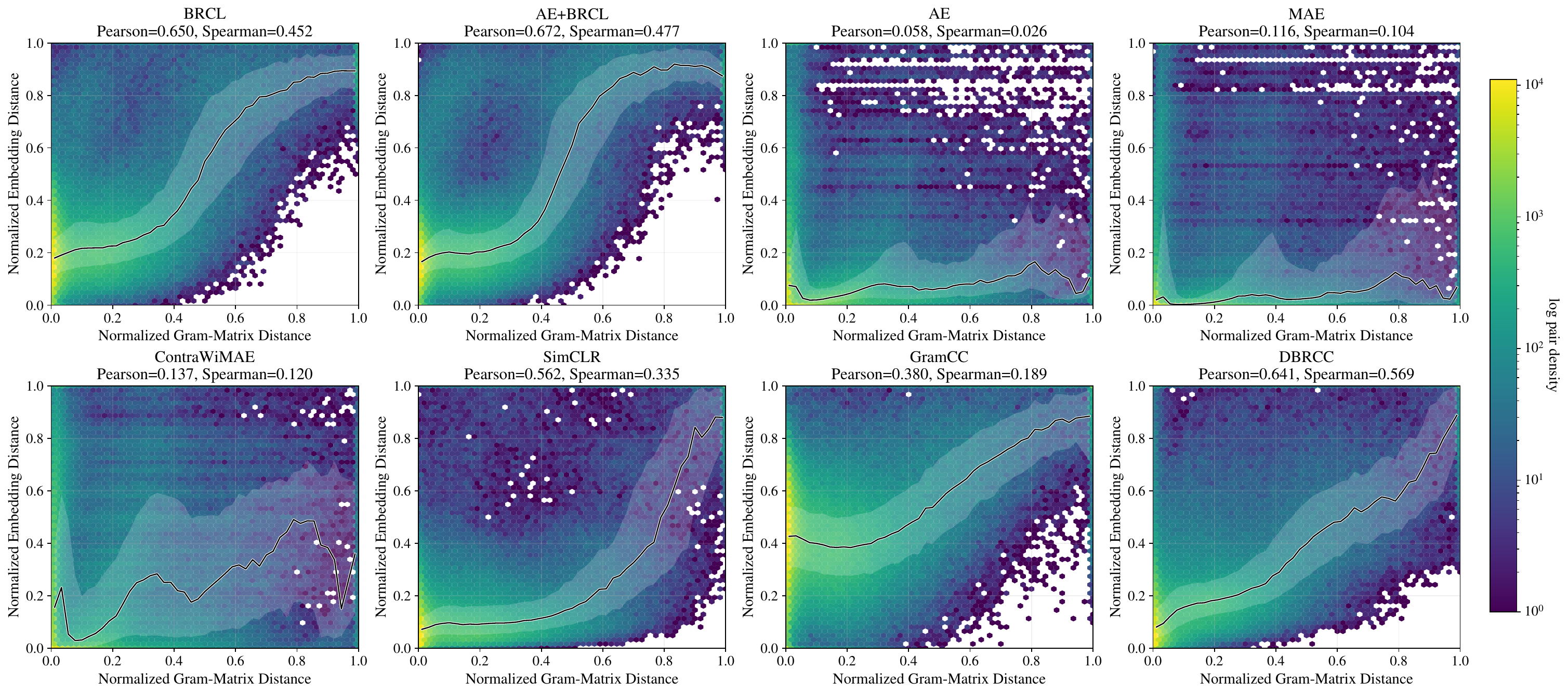}
    \caption{Pairwise alignment between Gram-matrix-domain and embedding-domain $\ell_2$ distances. Each hexbin plot shows normalized distances over sampled CSI pairs, with color indicating the log-scaled pair density. Median trends with interquartile ranges and Pearson/Spearman correlations are reported in each panel.}
    \label{fig_visual1}
    \vspace{-6pt}
\end{figure*}

\begin{lemma}[Normalized beam-response bound]
\label{lem:beam_response_bound}
For any collection of unit-norm transmit beams $\{\mathbf{w}_f\}_{f=1}^{N_f}$,
\begin{equation}
    \frac{1}{N_f}
    \sum_{f=1}^{N_f}
    \left|
    \bar{\gamma}_i[f;\mathbf{w}_f]
    -
    \bar{\gamma}_j[f;\mathbf{w}_f]
    \right|^2
    \leq
    D_{\mathrm{shape}}(i,j)
    \leq
    D_{\mathrm{BR}}(i,j).
    \label{eq:beam_response_preservation_bound}
\end{equation}
\end{lemma}

\noindent\emph{Proof:}
For each subcarrier $f$,
\begin{equation}
    \bar{\gamma}_i[f;\mathbf{w}_f]
    -
    \bar{\gamma}_j[f;\mathbf{w}_f]
    =
    \mathbf{w}_f^{H}
    \mathbf{A}_{ij}[f]
    \mathbf{w}_f.
\end{equation}
Since $\mathbf{A}_{ij}[f]$ is Hermitian,
\begin{equation}
    \left|
    \mathbf{w}_f^{H}
    \mathbf{A}_{ij}[f]
    \mathbf{w}_f
    \right|
    \leq
    \left\|
    \mathbf{A}_{ij}[f]
    \right\|_2
    \leq
    \left\|
    \mathbf{A}_{ij}[f]
    \right\|_F.
    \label{eq:rayleigh_bound}
\end{equation}
Squaring and averaging over $f$ gives the first inequality. 
The second inequality follows from the definition of $D_{\mathrm{BR}}(i,j)$. \hfill $\blacksquare$

\subsection{Beam-Transfer Regret Bound}
\label{subsec:beam_decision_regret}

The Gram-matrix-induced dissimilarity is next related to codebook-based beam transfer.
The codebook is introduced only for interpretation; BRCL itself does not require a predefined beam codebook during pretraining.

Let $\mathcal{W}$ denote a nonempty finite set of unit-norm candidate transmit beams.
For sample $\mathcal{X}_i$ at subcarrier $f$, choose an optimal beam
\begin{equation}
    \mathbf{w}_i^{\star}[f]
    \in
    \arg\max_{\mathbf{w}\in\mathcal{W}}
    \bar{\gamma}_i[f;\mathbf{w}].
    \label{eq:best_beam_vector}
\end{equation}
Since $\mathcal{W}$ is finite, the maximizer exists.
The same argument also applies to a compact continuous beam set because $\bar{\gamma}_i[f;\mathbf{w}]$ is continuous in $\mathbf{w}$.

\begin{theorem}[Beam-transfer regret bound]
\label{thm:beam_decision_regret}
For two CSI samples $\mathcal{X}_i$ and $\mathcal{X}_j$, define the normalized beam-transfer regret incurred by using $\mathbf{w}_i^{\star}[f]$ on sample $\mathcal{X}_j$ as
\begin{equation}
    \mathcal{R}_{i\rightarrow j}[f]
    =
    \max_{\mathbf{w}\in\mathcal{W}}
    \bar{\gamma}_j[f;\mathbf{w}]
    -
    \bar{\gamma}_j[f;\mathbf{w}_i^{\star}[f]].
    \label{eq:beam_regret_definition}
\end{equation}
Then
\begin{equation}
    0
    \leq
    \mathcal{R}_{i\rightarrow j}[f]
    \leq
    2
    \left\|
    \mathbf{A}_{ij}[f]
    \right\|_2
    \leq
    2
    \left\|
    \mathbf{A}_{ij}[f]
    \right\|_F.
    \label{eq:beam_regret_bound}
\end{equation}
Consequently, the average squared regret over OFDM subcarriers satisfies
\begin{equation}
    \frac{1}{N_f}
    \sum_{f=1}^{N_f}
    \mathcal{R}_{i\rightarrow j}^2[f]
    \leq
    4D_{\mathrm{shape}}(i,j)
    \leq
    4D_{\mathrm{BR}}(i,j).
    \label{eq:average_beam_regret_bound}
\end{equation}
\end{theorem}

\noindent\emph{Proof:}
For any $\mathbf{w}\in\mathcal{W}$, the Rayleigh quotient bound gives
\begin{equation}
    \left|
    \bar{\gamma}_i[f;\mathbf{w}]
    -
    \bar{\gamma}_j[f;\mathbf{w}]
    \right|
    \leq
    \left\|
    \mathbf{A}_{ij}[f]
    \right\|_2.
    \label{eq:uniform_beam_bound}
\end{equation}
Let
\begin{equation}
    \delta_{ij}[f]
    =
    \left\|
    \mathbf{A}_{ij}[f]
    \right\|_2.
\end{equation}
Then
\begin{align}
    \mathcal{R}_{i\rightarrow j}[f]
    &=
    \max_{\mathbf{w}\in\mathcal{W}}
    \bar{\gamma}_j[f;\mathbf{w}]
    -
    \bar{\gamma}_j[f;\mathbf{w}_i^{\star}[f]]
    \nonumber\\
    &\leq
    \max_{\mathbf{w}\in\mathcal{W}}
    \bar{\gamma}_i[f;\mathbf{w}]
    -
    \bar{\gamma}_i[f;\mathbf{w}_i^{\star}[f]]
    +
    2\delta_{ij}[f]
    \nonumber\\
    &=
    2\delta_{ij}[f],
\end{align}
because $\mathbf{w}_i^{\star}[f]$ maximizes $\bar{\gamma}_i[f;\mathbf{w}]$ over $\mathcal{W}$.
This proves \eqref{eq:beam_regret_bound}. 
Squaring \eqref{eq:beam_regret_bound}, averaging over $f$, and using
$\|\mathbf{A}_{ij}[f]\|_2\leq\|\mathbf{A}_{ij}[f]\|_F$
give \eqref{eq:average_beam_regret_bound}. \hfill $\blacksquare$

\textit{Relation to OFDM-averaged beam selection:}
Although Theorem~\ref{thm:beam_decision_regret} is stated for subcarrier-wise beam decisions, Jensen's inequality implies that, for any fixed unit-norm codebook beam, the squared discrepancy between the corresponding OFDM-averaged normalized beam responses is also bounded by $D_{\mathrm{shape}}(i,j)\leq D_{\mathrm{BR}}(i,j)$.
Therefore, this response-level bound is consistent with the OFDM-averaged codebook beam-selection metric used in the experiments, while it should not be interpreted as a guarantee on classification accuracy, rate ratio, MU-MIMO scheduling utility, or temporal prediction performance.

\subsection{Implication for BRCL}
\label{subsec:theory_implication}

Lemma~\ref{lem:beam_response_bound} and Theorem~\ref{thm:beam_decision_regret} provide a physical rationale for the Gram-matrix-induced target used in BRCL.
If $D_{\mathrm{BR}}(i,j)$ is small, then the normalized transmit responses of $\mathcal{X}_i$ and $\mathcal{X}_j$ are close in the sense of the bounds above.
In particular, transferring a codebook beam selected for one sample to the other incurs a bounded normalized beam-response regret.

This interpretation justifies using $D_{\mathrm{BR}}(i,j)$ to construct soft positive relations during contrastive pretraining.
The instance-level loss preserves augmentation-based same-sample invariance, whereas the Gram-matrix-relational loss encourages the embedding space to reflect transmitter-side beam-response similarity among distinct CSI samples.
Thus, the analysis provides a physical justification of the proposed pretraining target, rather than an end-to-end guarantee on downstream task performance.
It shows that the self-supervised relation is aligned with a fundamental transmit beam-response property of MIMO channels.

Fig.~\ref{fig:dbr_regret_bound} provides an empirical consistency check for this interpretation.
For both global random CSI pairs and local $D_{\mathrm{BR}}$-neighbor pairs, the measured average squared beam-transfer regret remains below the theoretical upper bound.
Moreover, the regret statistics tend to increase with $D_{\mathrm{BR}}$, indicating that CSI pairs with smaller beam-response dissimilarity are more likely to exhibit lower beam-transfer regret in the generated channel population.



\section{Experiments}
\label{sec:experiments}

\subsection{Dataset and Simulation Setup}
Following prior work, ray-tracing-based CSI is generated using the DeepMIMO dataset framework~\cite{ALKHATEEB2019DEEPMIMO}, which enables fair comparison through a common data-generation pipeline, facilitates reproducibility through its open framework, and supports geographical generalization evaluation across diverse ray-traced urban environments.

The training set consists of 2.5 million CSI samples generated at 3.5~GHz from 56 city scenarios and is used for both self-supervised pretraining and downstream-head training.\footnote{The training scenarios are \textit{Amsterdam, ASU Campus, Athens, Bangkok, Barcelona, Beijing, Boston, Brussels, Cairo, Cape Town, Charlotte, Chicago, Denver, Dubai, Edinburgh, Florence, Fort Worth, Fujiyoshida, Granada, Hatsukaichi, Helsinki, Hong Kong, Indianapolis, Istanbul, Jerusalem, Kyoto, Havana, Lisbon, Los Angeles, Madrid, Marrakesh, Mumbai, New Delhi, New York, North Jakarta, Oklahoma City, Philadelphia, Phoenix, Reykjavik, Rio de Janeiro, Rome, San Francisco, San Nicolas, Saint Petersburg, Santa Clara, Santiago, Seattle, Seoul, Singapore, Stockholm, Sumida City, Sydney, Taipei, Taito City, Toronto, and Warsaw.}}
Downstream testing uses a separate test-only set of 0.55 million CSI samples from 10 unseen scenarios disjoint from the training scenarios.\footnote{The unseen test scenarios are \textit{Austin, Centro, Columbus, Dallas, Gurbchen, Houston, Miami, Montreal, Prague, and San Diego.}}
Unless otherwise stated, the MIMO-OFDM system uses $N_f=32$, $\Delta f=30$~kHz, 960-kHz bandwidth, $N_t=32$ half-wavelength ULA BS antennas, $N_r=4$ receive antennas, and up to 20 dominant paths.
For temporal tasks, CSI sequences are sampled along mobility trajectories at one-coherence-time intervals, $T_c=13.0$~ms, corresponding to a maximum speed of approximately 10~km/h at 3.5~GHz.

\subsection{Evaluation Protocol and Baselines}
\label{subsec:evaluation_protocol}

Unless otherwise specified, the following evaluation protocol is used throughout the experiments.
Task-specific heads are trained using labeled samples from the 56 training scenarios and evaluated on the 10 unseen test scenarios, yielding a cross-scenario transfer setting. To assess label efficiency, only a fraction of the downstream training labels is used, with labeled-data ratios $\rho_{\mathrm{label}}\in\{1\%,5\%,10\%,20\%,100\%\}$.
For each pretrained encoder, the encoder parameters are frozen, and only the task-specific modules are trained using the available labeled samples.
Three representative transmitter-side tasks are considered: beam selection, MU-MIMO user selection, and future beam selection.
All experiments are conducted over datasets generated using five independent random user-location seeds, and the reported metrics are averaged over the resulting datasets.

All methods use a 256-dimensional representation and are configured with approximately 3M trainable parameters.
For CNN-based methods, the BRCL variants and CNN-based baselines share the same 16-layer CNN backbone.
The compared pretraining methods are summarized as follows.

\vspace{0.3em}
\noindent\textbf{BRCL}: The proposed method trained with instance-level and Gram-matrix-relational losses, without reconstruction regularization. The hyperparameters are set to $\alpha=0.9$, $\lambda_g=0.5$, $\beta_{\mathrm{rec}}=0$, $\tau_p=0.05$, and $\tau_z=0.2$.

\vspace{0.3em}
\noindent\textbf{AE+BRCL}: A reconstruction-regularized variant of BRCL that adds an $\ell_2$ reconstruction loss to the BRCL objective. The hyperparameters are set to $\alpha=0.9$, $\lambda_g=0.5$, and $\beta_{\mathrm{rec}}=0.5$.

\vspace{0.3em}
\noindent\textbf{AE}: An autoencoder-based CSI pretraining baseline following CsiNet~\cite{WEN2028CSINET}, trained only with reconstruction loss.

\vspace{0.3em}
\noindent\textbf{MAE}: A masked autoencoder baseline for CSI representation learning~\cite{WANG2025MASKED}. It uses an 8-layer Transformer encoder with $4\times4$ patches and a masking ratio of 0.75.

\vspace{0.3em}
\noindent\textbf{ContraWiMAE}: A contrastive masked autoencoder baseline~\cite{GULER2025MASKED}, using the same Transformer configuration as MAE.

\vspace{0.3em}
\noindent\textbf{SimCLR}: An instance-level contrastive learning baseline that forms positive pairs from two augmented views of the same CSI sample and treats other mini-batch samples as negatives~\cite{CHEN2020SIMCLR,GULER2025MASKED}.

\vspace{0.3em}
\noindent\textbf{GramCC}: A channel-charting-based baseline that defines inter-sample relations using a Gram-matrix-based similarity metric~\cite{STUDER2018CCHT}.

\vspace{0.3em}
\noindent\textbf{DBRCC}: A channel-charting-inspired baseline that defines inter-sample relations using the proposed beam-response dissimilarity $D_{\mathrm{BR}}$.

\subsection{Representation Analysis}
\label{subsec:representation_analysis}

Fig.~\ref{fig_visual1} quantifies how well the learned embedding distances preserve the Gram-matrix-domain geometry of CSI samples.
BRCL and its reconstruction-regularized variant, AE+BRCL, achieve the highest Pearson correlations, indicating that the proposed objectives better preserve the metric-scale structure induced by Gram-matrix-based channel similarity, which is consistent with their strong performance in the downstream tasks.
Meanwhile, DBRCC achieves the highest Spearman correlation, showing that the proposed beam-response dissimilarity $D_{\mathrm{BR}}$ provides a more effective rank-order similarity measure than the plain Gram-matrix $\ell_2$ distance.


\begin{figure}[h]
\centering
    \includegraphics[width=0.94\linewidth]{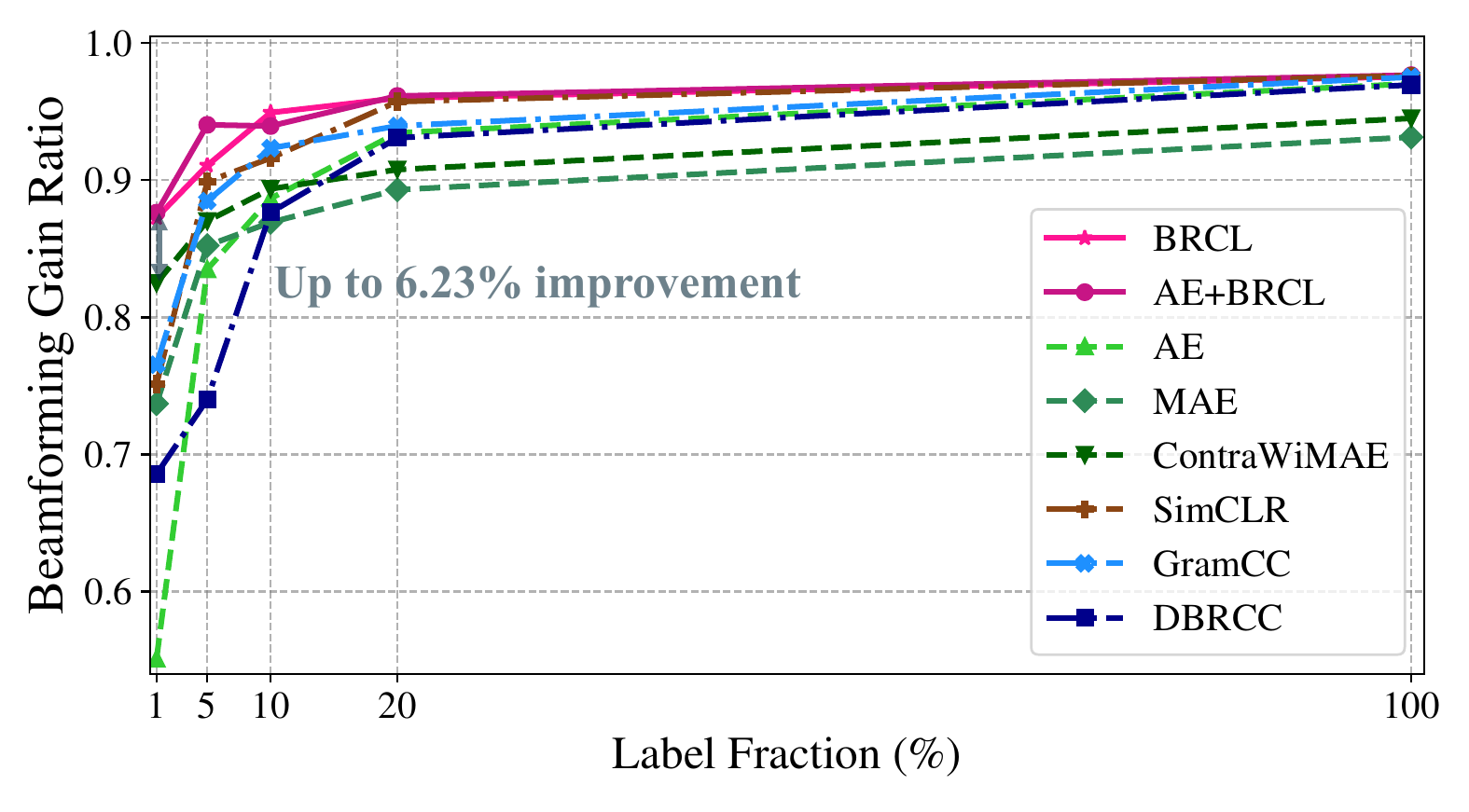}
    \vspace{-5pt}
    \caption{
    Beamforming gain ratio for beam selection in unseen scenarios under different labeled-data fractions. Results are averaged over five independent user-location seeds.
    }
    \vspace{-8pt}
    \label{fig:DS1}
\end{figure}

\subsection{Downstream Task 1: Beam Selection in Unseen Scenarios}
\label{subsec:downstream_beam_selection}

The first downstream evaluation considers out-of-scenario beam selection.
This task directly measures whether the pretrained encoder captures transmitter-side information that is useful for beam management in propagation environments not observed during self-supervised pretraining.

Let \(\mathcal{W}=\{\mathbf w_b\}_{b=1}^{B}\) denote a ULA-DFT transmit codebook with \(B\) beams, where \(\mathbf w_b\in\mathbb{C}^{N_t\times 1}\) and \(\|\mathbf w_b\|_2=1\).
For each CSI sample \(X_j=\{\mathbf H_j[f]\}_{f=1}^{N_f}\), the codebook-domain beam gain of beam \(b\) is computed from the raw CSI as
\[
    g_b(j)
    =
    \frac{1}{N_f}
    \sum_{f=1}^{N_f}
    \left\|
        \mathbf H_j[f]\mathbf w_b
    \right\|_2^2 .
\]
The ground-truth beam label is then defined as
\[
    b_j^\star
    =
    \arg\max_{1\leq b\leq B} g_b(j).
\]
This label-generation rule depends only on the codebook-domain transmit-beam response and does not require SNR calibration.
Input normalization is applied to the encoder input.

For representation-based prediction, each CSI sample is passed through the pretrained encoder.
A linear classification head is then trained on top of the representation to predict the beam logits over the $B$ codebook beams.
The predicted beam is selected as the beam with the largest output logit, and the head is trained using the cross-entropy loss with the ground-truth beam label ($b_j^\star$).

This downstream task is evaluated using a \(B=32\) beam codebook under different labeled-data fractions.
The beamforming gain ratio is defined as the beamforming gain achieved by the predicted beam normalized by that of the optimal beam.
As shown in Fig.~\ref{fig:DS1}, BRCL and AE+BRCL achieve the highest or near-highest gain ratios across all labeled-data fractions.
The advantage is particularly pronounced in the label-scarce regime; with only 1\% labeled data, AE+BRCL improves the beamforming gain ratio by approximately 6.23\% over the strongest non-BRCL baseline.
These results demonstrate that beam-response-aware pretraining yields label-efficient and transferable representations for beam management in unseen scenarios.


\begin{figure}[t]
\captionsetup[subfloat]{farskip=2pt}
\centering
        \subfloat[Spearman rank correlation]{\includegraphics[width=.95\linewidth]{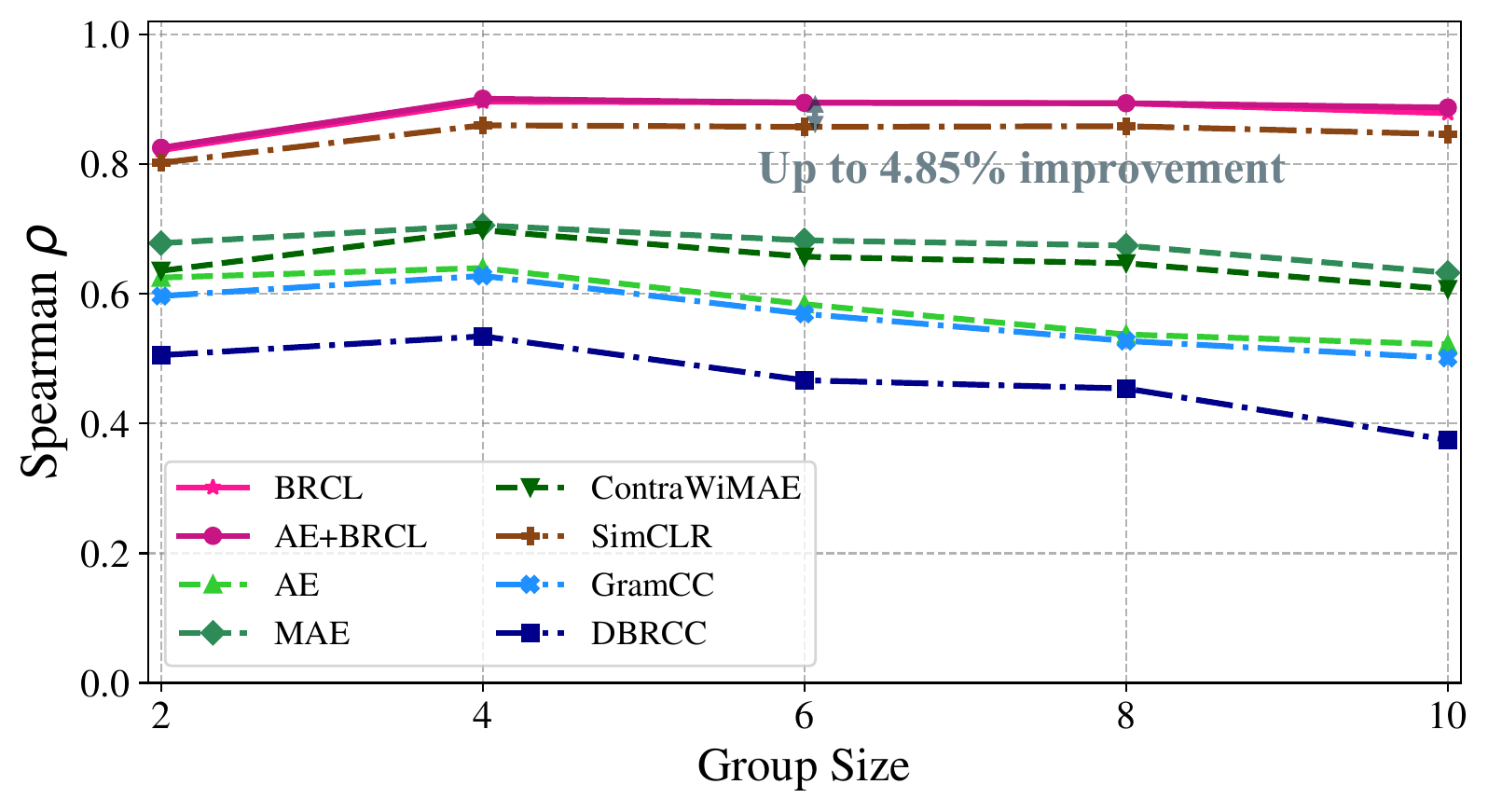}}\\
        \subfloat[Top sum-rate ratio]{\includegraphics[width=.95\linewidth]{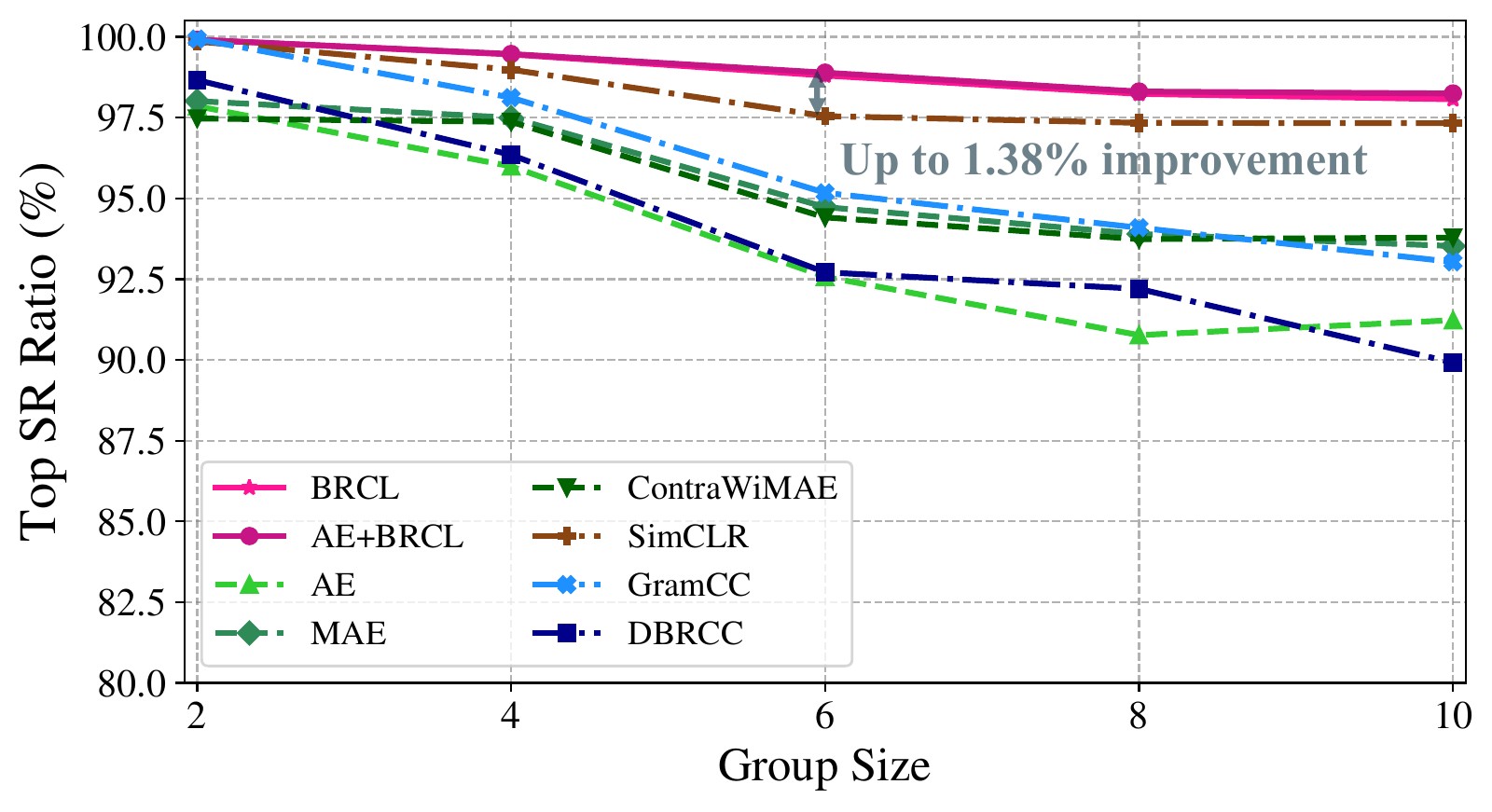}}
        \\
\vspace{-3pt}
\caption{MU-MIMO user selection performance as a function of group size. The SNR is $10$ dB, and results are averaged over labeled-data fractions and five independent user-location seeds. Higher Spearman correlation and sum-rate ratio indicate better grouping quality.
}
\vspace{-7pt}
\label{fig:DS2_2}
\end{figure}

\begin{figure*}[ht!]
\captionsetup[subfloat]{farskip=2pt}
\vspace{-7pt}
\centering
        \subfloat[$\rho_{\mathrm{label}} = 1\%$]{\includegraphics[width=.33\linewidth]{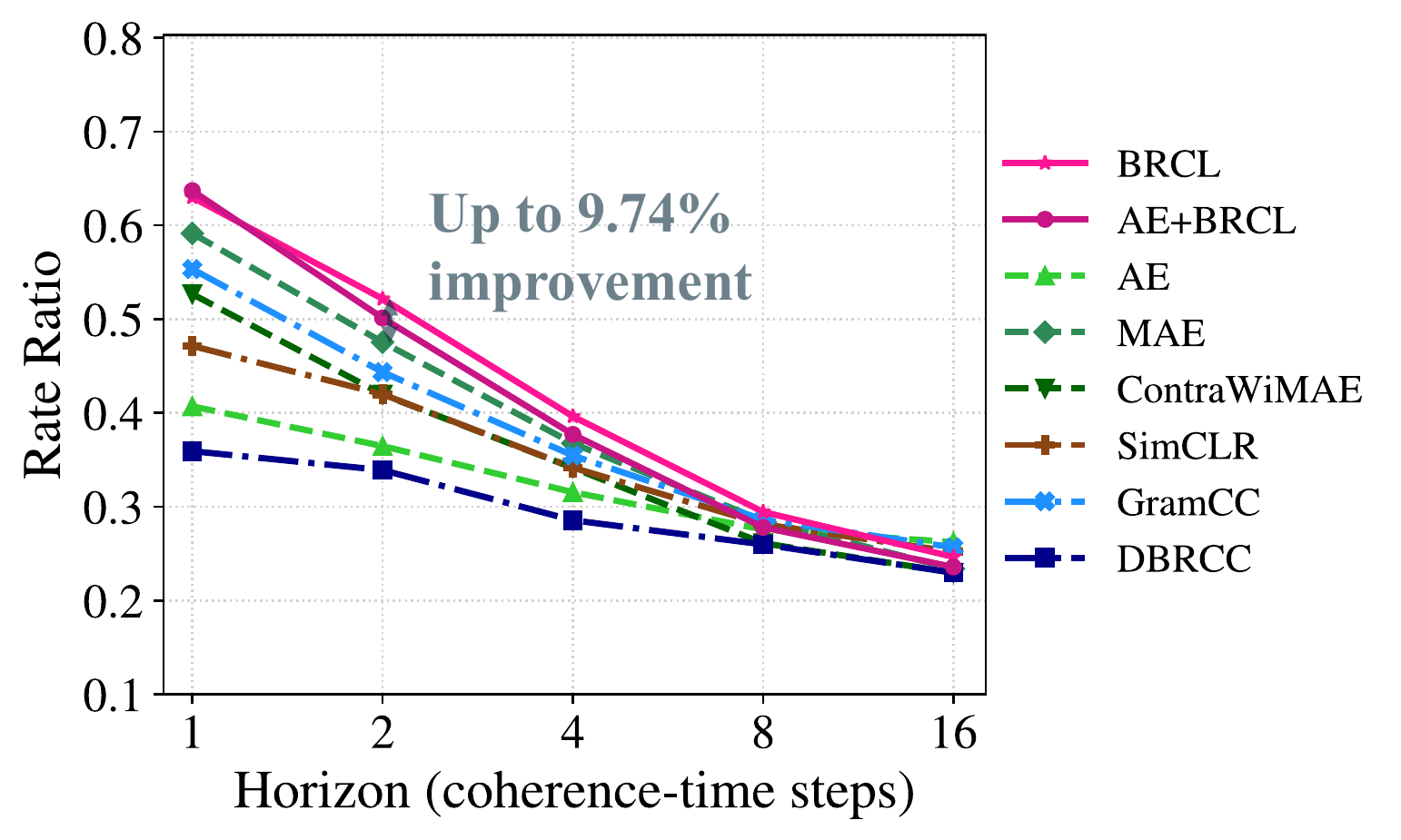}}
        \subfloat[$\rho_{\mathrm{label}} = 5\%$]{\includegraphics[width=.33\linewidth]{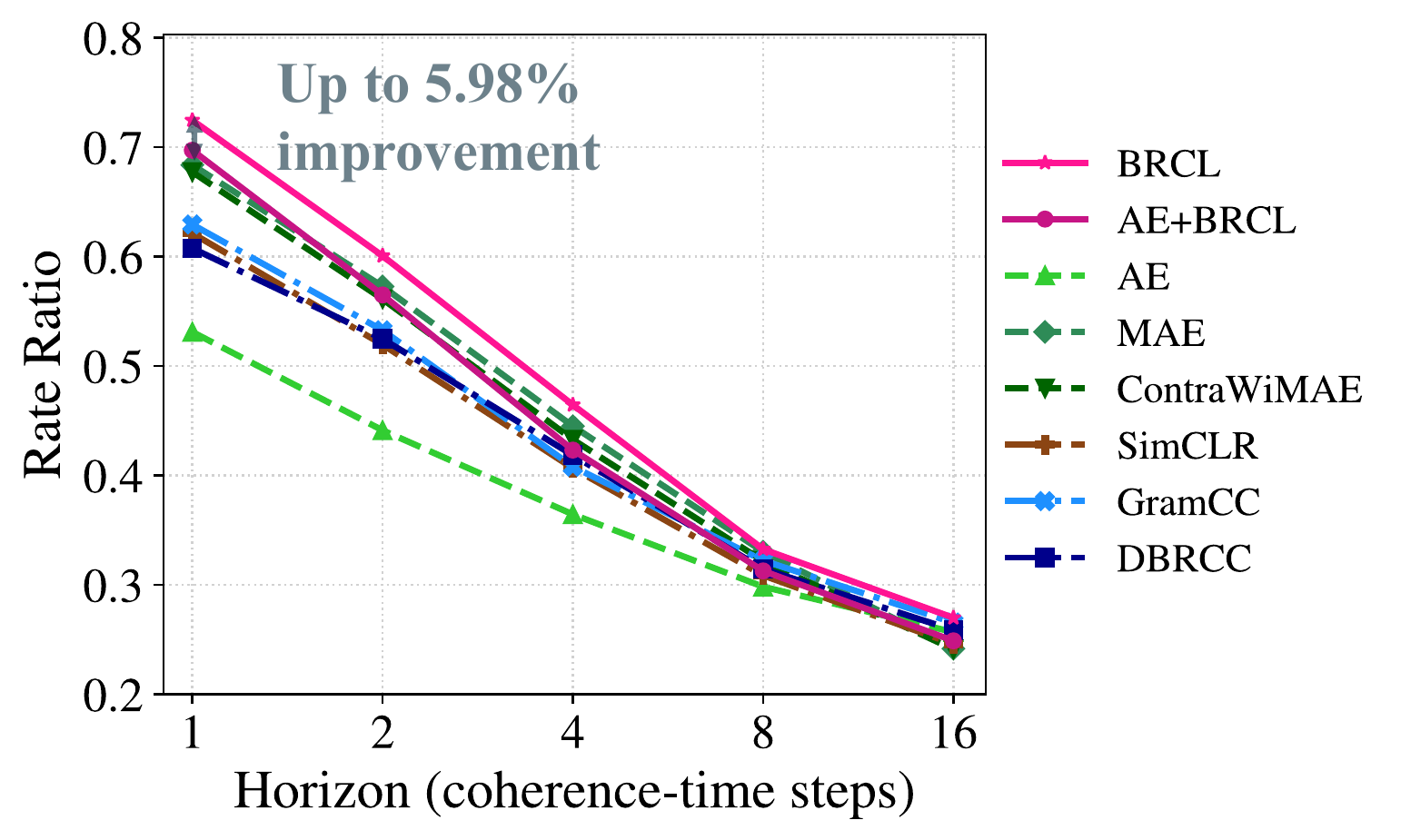}}
        \subfloat[$\rho_{\mathrm{label}} = 10\%$]{\includegraphics[width=.33\linewidth]{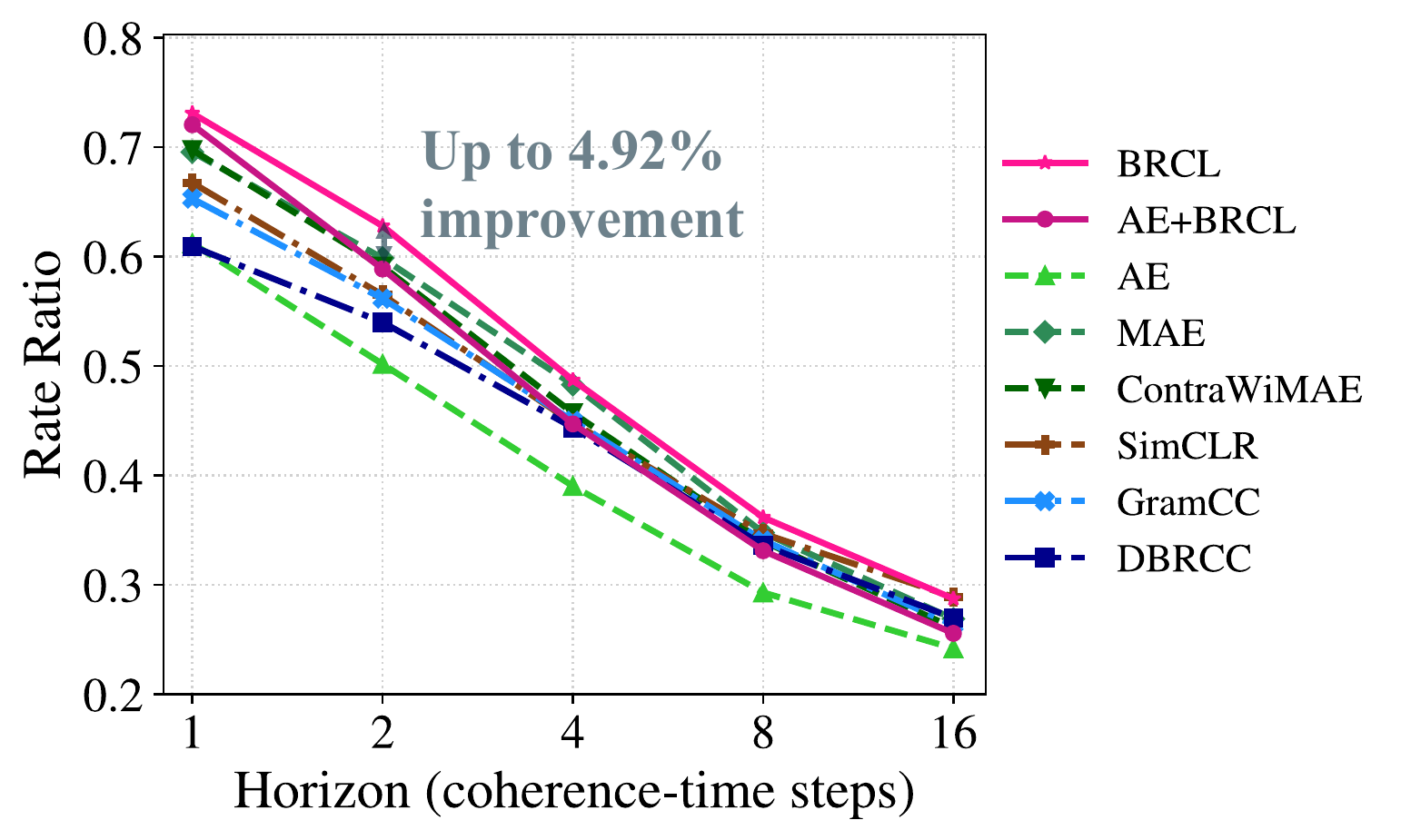}}\\
        \subfloat[$\rho_{\mathrm{label}} = 20\%$]{\includegraphics[width=.33\linewidth]{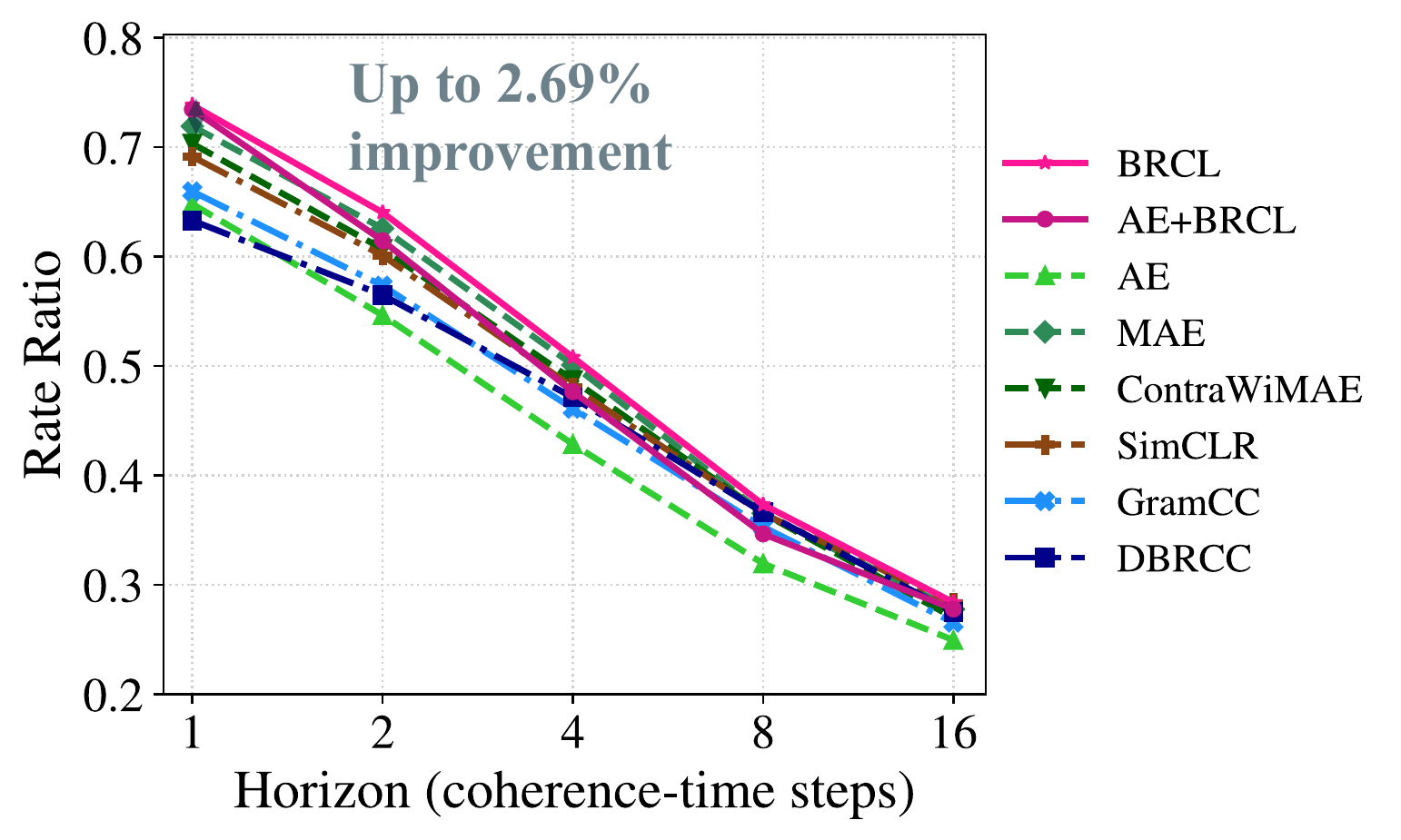}}
        \subfloat[$\rho_{\mathrm{label}} = 100\%$]{\includegraphics[width=.33\linewidth]{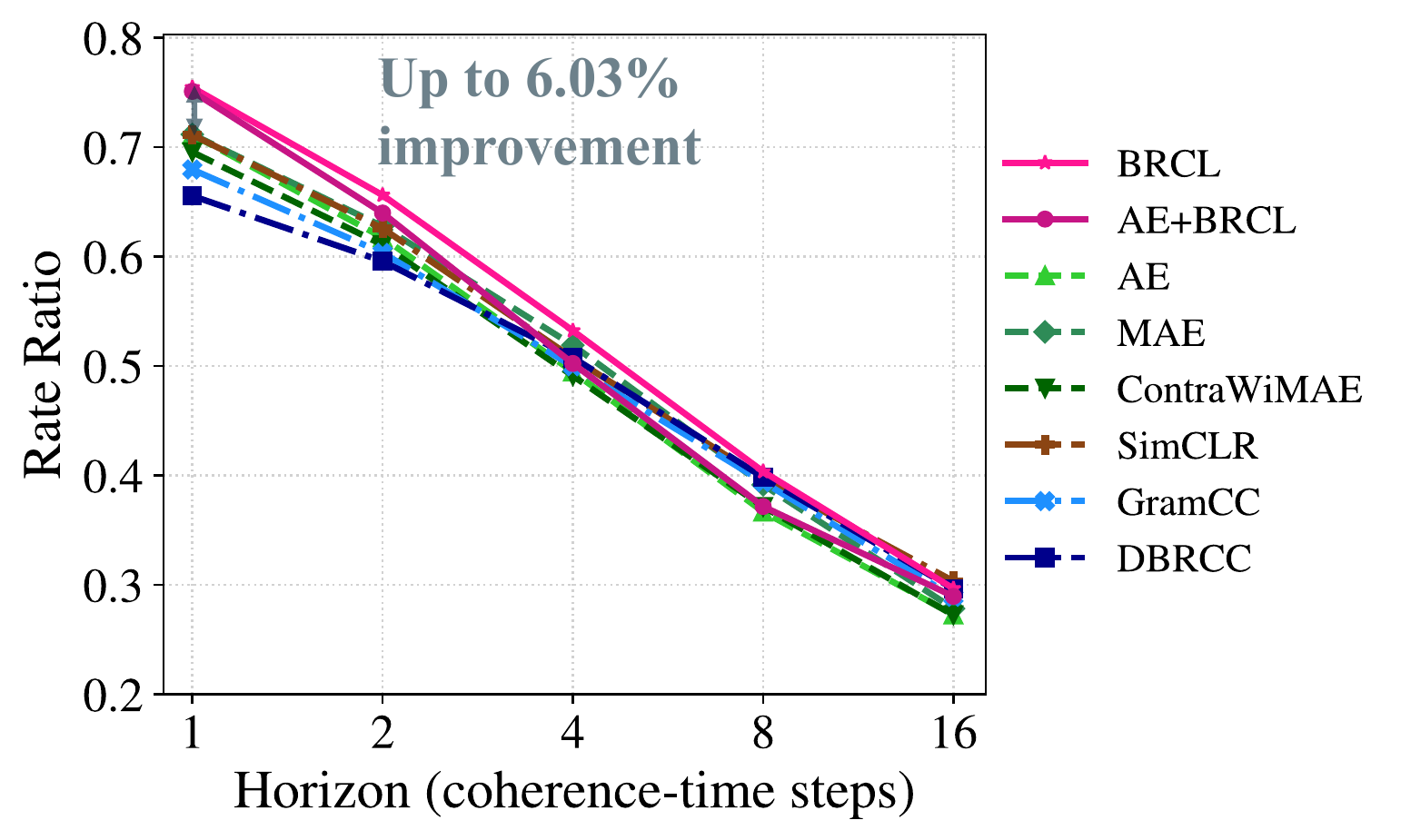}}
        \\
\vspace{-3pt}
\caption{
Future beam selection performance across prediction horizons and labeled-data fractions, measured by rate ratio. Results are averaged over five independent user-location seeds.
}
\vspace{-2pt}
\label{fig:DS3}
\end{figure*}

\subsection{Downstream Task 2: MU-MIMO User Selection}
\label{subsec:downstream_mu_mimo_selection}

The next downstream task considers multi-user MIMO (MU-MIMO) user selection~\cite{YOO2006UE,YOO2011UE}.
Unlike single-user beam selection, MU-MIMO scheduling depends on the joint spatial separability of multiple users.
A high-quality user set should contain users whose channels can be simultaneously served with limited inter-user interference after precoding.
This task therefore examines whether the pretrained encoder preserves the multi-user spatial structure required for transmitter-side scheduling, beyond individual beam preferences.

\begin{table}[t]
\centering
\caption{Averaged performance comparison of different representation learning schemes for the MU-MIMO user selection downstream task over five independent user-location seeds.}
\label{tab:ds2_metrics}
\resizebox{0.88\linewidth}{!}{
\begin{tabular}{lcccc}
\toprule[1pt]
Method 
& \shortstack{Top SR\\Ratio (\%) $\uparrow$} 
& \shortstack{Spearman\\$\rho$ $\uparrow$} 
& \shortstack{Good-group\\AUC $\uparrow$} 
& \shortstack{Top\\Recall $\uparrow$} \\
\midrule
BRCL 
& \underline{98.90} & \underline{0.877} & \underline{\textbf{0.927}} & \underline{0.614} \\
AE+BRCL 
& \underline{\textbf{98.97}} & \underline{\textbf{0.881}} & \underline{0.924} & \underline{\textbf{0.617}} \\
AE 
& 93.68 & 0.582 & 0.829 & 0.403 \\
MAE 
& 95.53 & 0.675 & 0.865 & 0.468 \\
ContraWiMAE 
& 95.36 & 0.649 & 0.858 & 0.469 \\
SimCLR 
& 98.21 & 0.845 & 0.895 & 0.544 \\
GramCC 
& 96.07 & 0.564 & 0.826 & 0.452 \\
DBRCC 
& 93.97 & 0.467 & 0.775 & 0.372 \\
\bottomrule[1pt]
\end{tabular}
}
\vspace{2pt}
\begin{minipage}{0.99\linewidth}
\small
*Best: \underline{\textbf{bold}}, second-best: \underline{underline}.
\end{minipage}
\vspace{-9pt}
\end{table}

Candidate user-index sets are constructed from unseen scenarios as
\begin{equation}
    \mathcal{U}
    =
    \{j_1,j_2,\ldots,j_{K_{\mathrm{u}}}\},
\end{equation}
where $K_{\mathrm{u}}$ denotes the number of co-scheduled users. The corresponding CSI samples are denoted by
\begin{equation}
    \mathcal{X}_{\mathcal{U}}
    =
    \{\mathcal{X}_{j_k}\}_{k=1}^{K_{\mathrm{u}}}.
\end{equation}
The number of selected users is set to $K_{\mathrm{u}}\in\{2,4,6,8,10\}$ in the experiments. Users in each candidate set are sampled from the same scenario so that the candidates represent physically meaningful MU-MIMO scheduling groups.

The ground-truth utility of each candidate user set is computed using regularized zero-forcing (RZF) precoding~\cite{PEEL2005RZF}. To focus the utility label on spatial compatibility rather than large-scale channel gain, each user's channel is normalized before utility computation. For each user $j$, a single-stream effective transmit-side channel is formed by applying the dominant receive combining vector $\mathbf{c}_j[f]\in\mathbb{C}^{N_r}$ to the MIMO channel $\mathbf{H}_j[f]\in\mathbb{C}^{N_r\times N_t}$:
\begin{equation}
    \mathbf{h}_j[f]
    =
    \mathbf{c}_j[f]^H
    \mathbf{H}_j[f]
    \in
    \mathbb{C}^{1\times N_t}.
\end{equation}
The normalized effective channel is denoted by $\bar{\mathbf{h}}_j[f]$. For each candidate set $\mathcal{U}$, the normalized effective channels are stacked at each subcarrier, and an RZF precoder with column normalization is used to evaluate the achievable sum-rate. The resulting average sum-rate over OFDM subcarriers is used as the utility label. Unless otherwise stated, the utility is evaluated at a fixed receiver-side SNR of $10$ dB after per-user normalization.

For representation-based prediction, the pretrained encoder is shared across all users in the candidate set.
Each CSI sample $\mathcal{X}_{j_k}$ is encoded into a user representation, and the resulting set of representations is fed into a permutation-invariant set-level prediction head.
The head outputs both a scheduling utility score and a high-quality user-set probability.

For each $K_{\mathrm{u}}$, candidate sets in the top $25\%$ of the RZF-based utility distribution among the corresponding training candidates are labeled as high-quality sets.
The set-level head is trained with a weighted combination of utility regression, high-quality-set classification, and an optional spatial-compatibility regularizer based on effective transmit-direction overlap.
This objective promotes representations that support utility ranking, high-quality set identification, and interference-aware user separability.
Fig.~\ref{fig:DS2_2} shows the MU-MIMO user selection performance under different group sizes.
BRCL and AE+BRCL consistently achieve strong performance, with clearer advantages for larger user groups, where accurate estimation of spatial compatibility becomes more challenging.
Table~\ref{tab:ds2_metrics} further shows that the BRCL variants consistently rank first or second across all metrics, with AE+BRCL leading the top sum-rate ratio, Spearman correlation, and top recall, and BRCL achieving the best good-group AUC.
These results indicate that beam-response-aware pretraining preserves CSI representations that transfer effectively to interference-aware MU-MIMO user selection under limited labeled supervision.

\subsection{Downstream Task 3: Future Beam Selection}
\label{subsec:downstream3_channel_prediction}

The learned CSI representation is further evaluated for future beam selection to assess whether it captures temporally predictive channel information.
While Downstream Task 1 considers beam selection from the current CSI sample, this task extends the same codebook-domain criterion to a temporal setting:
given a short history of past CSI snapshots, the model predicts the future optimal transmit beam at multiple prediction horizons.
This formulation evaluates whether the pretrained representation preserves channel information that remains useful for anticipating future transmitter-side decisions, without requiring full complex CSI reconstruction.

Let \(L\) denote the number of past CSI snapshots used as input and let \(\Delta\) denote the temporal spacing between consecutive snapshots.
Given the CSI history $ \{X_{t-(L-1)\Delta},\ldots,X_{t-\Delta},X_t\},$ the model predicts the optimal codebook beam at future horizons \(q\in\mathcal Q\).
Using the same ULA-DFT transmit codebook \(\mathcal W=\{\mathbf w_b\}_{b=1}^{B}\) as in the beam-selection task, the future beam gain at horizon \(q\) is computed as
\[
    g_b(t+qT_c)
    =
    \frac{1}{N_f}
    \sum_{f=1}^{N_f}
    \left\|
        \mathbf H_{t+qT_c}[f]\mathbf w_b
    \right\|_2^2,
\]
and the corresponding target beam is
\[
    b^{\star}(t+qT_c)
    =
    \arg\max_{1\leq b\leq B} g_b(t+qT_c),
    \qquad q\in\mathcal Q.
\]
Here, $T_c$ denotes the coherence time, such that $q$ represents the prediction horizon in units of $T_c$.
The experiments use $\mathcal{Q}={1,2,4,8,16}$, a history length of $L=4$, and a beam codebook of size $B=64$.
In the test set, these horizons correspond to average look-ahead times of approximately \(13.0\), \(26.1\), \(52.2\), \(104.3\), and \(208.7\) ms, respectively.
Rate-oriented metrics are evaluated at an SNR of \(10\) dB.

For prediction, each CSI snapshot in the history is first passed through the pretrained encoder to obtain a sequence of representations.
A gated recurrent unit (GRU) \cite{CHO2014GRU} is then used as a lightweight temporal adapter to extract short-term temporal dependencies from the encoded CSI history.
The resulting temporal state is passed to an MLP head that outputs future beam logits for all horizons in \(\mathcal Q\).

Fig.~\ref{fig:DS3} shows the future beam selection performance across prediction horizons and labeled-data fractions.
BRCL achieves the best performance for horizons up to 8 under all labeled-data fractions, with a particularly clear advantage in the label-scarce regime.
Although the rate ratio generally improves as more labels are used, the cross-scenario setting still limits performance even with full labeled data.
These results indicate that BRCL learns temporally informative CSI representations that transfer effectively to future beam-domain decisions, especially when labeled future-beam data are scarce.

\begin{table}[t]
\centering
\caption{
Ablation correlation results between embedding-space $\ell_2$ distance and $D_{\mathrm{BR}}$ on a held-out validation split. Unless swept, default parameters are $\alpha=0.9$, $\lambda_g=0.5$, $\beta_{\mathrm{rec}}=0.5$, $K=16$, $\tau_p=0.05$, $\tau_z=0.2$, $d_z=512$, and batch size 256.
}
\label{tab:ablations}
\resizebox{0.64\linewidth}{!}{
\begin{tabular}{cccc}
\toprule
Parameter & Value & Pearson $\uparrow$ & Spearman $\rho$ $\uparrow$ \\
\midrule
\multirow[c]{5}{*}{$\alpha$}
& 0.1 & 0.6352 & 0.4402 \\
& 0.3 & 0.6228 & 0.4006 \\
& 0.5 & 0.5788 & 0.3589 \\
& 0.7 & \underline{0.6370} & \textbf{\underline{0.4506}} \\
& 0.9 & \textbf{\underline{0.6401}} & \underline{0.4494} \\
\midrule
\multirow[c]{5}{*}{$\lambda_g$}
& 0.1 & 0.5723 & 0.3382 \\
& 0.3 & 0.6235 & \textbf{\underline{0.3853}} \\
& 0.5 & \textbf{\underline{0.6507}} & 0.3641 \\
& 0.7 & \underline{0.6392} & \underline{0.3676} \\
& 0.9 & 0.5962 & 0.3599 \\
\midrule
\multirow[c]{5}{*}{$\beta_{\mathrm{rec}}$}
& 0.1 & 0.6483 & 0.4334 \\
& 0.3 & 0.6806 & \textbf{\underline{0.5210}} \\
& 0.5 & \underline{0.6823} & \underline{0.5024} \\
& 0.7 & \textbf{\underline{0.6881}} & 0.4486 \\
& 0.9 & 0.6234 & 0.4595 \\
\bottomrule
\end{tabular}
}
\vspace{2pt}
\begin{minipage}{0.94\linewidth}
\:\:\:\:\:\; \small *Best: \underline{\textbf{bold}}, second-best: \underline{underline}.
\end{minipage}
\vspace{-8pt}
\end{table}

\subsection{Ablation Study}
\label{subsec:ablation}

Table~\ref{tab:ablations} examines the sensitivity of BRCL to the main loss weights $\alpha$, $\lambda_g$, and $\beta_{\mathrm{rec}}$ on a held-out validation split.
Higher Pearson and Spearman correlations indicate better preservation of the metric and rank-order structures induced by $D_{\mathrm{BR}}$.

The results show that the relational weight $\alpha$ is important for aligning the embedding geometry with the beam-response dissimilarity.
The gain weight $\lambda_g$ controls the contribution of absolute channel-gain information, while the reconstruction weight $\beta_{\mathrm{rec}}$ balances raw CSI recovery and beam-response-aware geometry preservation.
These results motivate the use of $\alpha=0.9$ and $\lambda_g=0.5$, while $\beta_{\mathrm{rec}}=0.5$ is adopted as a balanced reconstruction weight in the main experiments.

\section{Conclusion}
\label{sec:conclusion}



This paper proposed beam-response contrastive learning (BRCL), a self-supervised CSI representation learning framework for transmitter-side MIMO tasks.
Unlike reconstruction-based or augmentation-driven pretraining, BRCL defines inter-sample similarity through the transmit-side Gram matrix, which determines the channel response to transmit beams.
By converting the resulting beam-response dissimilarity into a soft relational target, BRCL learns CSI embeddings that preserve precoding-relevant channel geometry without task-specific labels.

This work showed that the proposed beam-response dissimilarity admits a physical interpretation through its connection to transmit beam-response discrepancy and beam-transfer regret.
Experiments on large-scale ray-tracing MIMO-OFDM datasets demonstrated that BRCL improves representation alignment, label efficiency, and downstream performance across beam selection, MU-MIMO user selection, and future beam selection in unseen scenarios.
These results underscore the potential of communication-oriented similarity as a foundation for transferable CSI representation learning and motivate future research on task-adaptive or universal optimal similarity metrics and QoE-driven objectives for AI-enabled MIMO systems.


\appendices

\bibliographystyle{IEEEtran}
\bibliography{ref}

\end{document}